\documentclass[twocolumn,a4paper,accepted=2024-09-12]{quantumarticle}
\pdfoutput=1
\usepackage[utf8]{inputenc}
\usepackage[english]{babel}
\usepackage[T1]{fontenc}
\usepackage{amsmath}
\usepackage{hyperref}
\usepackage{amssymb}
\usepackage{tikz}
\usepackage{amsthm}
\usepackage{cleveref}

\usepackage[numbers,sort&compress]{natbib}

\newtheorem{theorem}{Theorem}
\newtheorem{corollary}{Corollary}[theorem]
\newtheorem{lemma}[theorem]{Lemma}
\theoremstyle{definition}

\usepackage[normalem]{ulem}

\newcommand{\iu}{{i\mkern1mu}}
\newcommand\scalemath[2]{\scalebox{#1}{\mbox{\ensuremath{\displaystyle #2}}}}
\newcommand{\ket}[1]{\left|#1\right\rangle}
\newcommand{\bra}[1]{\left\langle #1\right|}

\begin{document}

\title{Adaptive Phase Estimation with Squeezed Vacuum Approaching the Quantum Limit}

\author{M. A. Rodríguez-García}
\author{F. E. Becerra}
\affiliation{Center for Quantum Information and Control, Department of Physics and Astronomy, University of New Mexico, Albuquerque, New Mexico 87131, USA}


\maketitle

\begin{abstract}
Phase estimation plays a central role in communications, sensing, and
information processing. Quantum correlated states, such as squeezed states,
enable phase estimation beyond the shot-noise limit, and in principle approach
the ultimate quantum limit in precision, when paired with optimal quantum
measurements. However, physical realizations of optimal quantum measurements
for optical phase estimation with quantum-correlated states are still unknown.
Here we address this problem by introducing an adaptive Gaussian measurement
strategy for optical phase estimation with squeezed vacuum states that, by
construction, approaches the quantum limit in precision. This strategy builds
from a comprehensive set of locally optimal POVMs through rotations and
homodyne measurements and uses the Adaptive Quantum State Estimation framework
for optimizing the adaptive measurement process, which, under certain
regularity conditions, guarantees asymptotic optimality for this quantum
parameter estimation problem. As a result, the adaptive phase estimation
strategy based on locally-optimal homodyne measurements achieves the quantum
limit within the phase interval of $[0, \pi/2)$. Furthermore, we generalize
this strategy by including heterodyne measurements, enabling phase estimation
across the full range of phases from $[0, \pi)$, where squeezed vacuum allows
for unambiguous phase encoding. Remarkably, for this phase interval, which is
the maximum range of phases that can be encoded in squeezed vacuum, this
estimation strategy maintains an asymptotic quantum-optimal performance,
representing a significant advancement in quantum metrology.
\end{abstract}

\section{Introduction}

Quantum metrology uses the quantum properties of physical systems to enhance the
measurement precision of physical quantities beyond the classical limits
\cite{giovannetti2011advances, degen2017quantum}. Quantum mechanics states that
all physical observables are represented by self-adjoint operators on a Hilbert
space. As such, the measurement of a physical quantity of a system involves
projecting the quantum state of such system onto one of the eigenspaces of the
corresponding self-adjoint operator. However, certain physical quantities, such
as time, phase, or temperature, lack an associated self-adjoint operator
\cite{Hayashi2005,Holevo2011}. Consequently, to determine the values of these
physical quantities, it is necessary to measure some observables of the system
and estimate their values from the observed results. This process is referred to
as quantum parameter estimation \cite{Hayashi2005,Holevo2011}.

Among different parameter estimation problems, the problem of phase estimation
is ubiquitous in many areas of physics and engineering including, but not
limited to, gravitational wave detection \cite{aasi2013enhanced}, quantum
imaging \cite{gatti2008quantum}, atomic clocks \cite{kruse2016improvement},
magnetometry \cite{danilin2018quantum}, and quantum information processing
\cite{gilchrist2004schrodinger}. However, the performance of traditional phase
estimation methods is limited by the fundamental properties of the physical
states carrying the phase information. The maximum achievable precision for
phase estimation for probe states that lack quantum correlations, typically used
for phase estimation, is defined as shot-noise limit (SNL) \cite{
  giovannetti2011advances, Wiseman2009_book}.

Numerous methods have been developed for achieving phase estimation beyond the
SNL by exploiting probe states with inherent quantum correlations. Among
different types of quantum correlations, entanglement holds a significant
potential for improving precision in phase estimation. Nevertheless, highly
entangled states used for phase estimation, such as NOON states, are delicate
and can be readily disrupted by loss, environmental noise, and decoherence,
thereby limiting their practicality for real world applications
\cite{haroche1998entanglement,escher2011quantum,polino2020photonic,barbieri2022optical}.
In this regard, squeezed state probes offer a more viable alternative for robust
phase estimation \cite{caves1981quantum, Maccone2020squeezingmetrology}.
Squeezed states allow for reducing the quantum noise in one observable below the
SNL at the expense of an increased noise in another non-commuting observable.
This reduction in quantum noise can significantly enhance the precision of phase
measurements \cite{drummond2004quantum,aasi2013enhanced,kruse2016improvement,
  Maccone2020squeezingmetrology}, and is a valuable resource for enabling robust
optical quantum metrology and phase estimation.

Advances in photonic quantum technologies for phase estimation and quantum
metrology have yielded squeezed light sources with a high degree of squeezing
\cite{vahlbruch2007quantum, vahlbruch2016detection, schonbeck201813,
  PhysRevLett.128.083606}. Moreover, experimental demonstrations of quantum
metrology and sensing utilizing squeezed optical probes have achieved
sensitivities surpassing the SNL \cite{aasi2013enhanced, Berni2015,
  PhysRevLett.130.123603, lawrie2019quantum}. However, there remain significant
challenges in devising optimal estimation strategies, including optimal
measurements and estimators, that can efficiently attain the ultimate quantum
limit of precision for any optical phase estimation problem.

A noteworthy measurement approach for optical phase estimation with squeezed
states is the homodyne measurement. This measurement has the potential for
reaching the quantum limit for a specific, optimized phase when a predetermined
level of squeezing is present in the probe state \cite{Olivares2009}. However,
this optimal phase must be known beforehand in order to reach the quantum limit,
making this approach impractical. To overcome this limitation, two-step adaptive
methods for phase estimation allow for increasing the range of phases within
$[0,\pi/2)$ for which estimation below the SNL is possible
\cite{Olivares2009,Monras2006}. These methods can approximate the quantum limit
in precision in the asymptotic limit of many input states for phases around this
optimal phase. However, as the input phase deviates from the predetermined
optimal phase of homodyne, these estimation strategies show a considerable
discrepancy with the quantum limit.

In this work, we theoretically demonstrate a multi-step adaptive Gaussian
measurement strategy for optical phase estimation with squeezed vacuum states
that, by construction, approaches the quantum limit in precision with a fast
convergence rate for any phase encoded in squeezed vacuum. This estimation
strategy uses homodyne measurements to implement a comprehensive set of locally
optimal POVMs (Positive Operator Value Measures). Then the strategy performs
adaptive optimization based on the Adaptive Quantum State Estimation (AQSE)
framework to ensure the asymptotic consistency and efficiency of the estimator
of the optical phase \cite{Fujiwara2006}. Based on rigorous mathematical
analysis, we prove that this adaptive strategy approaches the quantum limit for
phases within $[0, \pi/2)$ in the asymptotic limit of many adaptive steps.
Furthermore, we generalize this strategy to incorporate heterodyne sampling
making it possible to extend the parametric range to $[0, \pi)$, which is the
maximum range of phases that can be encoded in squeezed vacuum. We show that
this combined homodyne-heterodyne strategy maintains an asymptotic quantum
optimal performance.

The paper is organized as follows: In Sec.~\ref{sec:II} we provide a concise
overview of the theory of single parameter estimation. Then, we discuss the
problem of optical phase estimation in the context of quantum systems, followed
by an overview of phase estimation with squeezed states. In
Sec.~\ref{sec:Op_dyne}, we describe the proposed optimal phase estimation
strategy based on adaptive Gaussian measurements with squeezed vacuum states. By
leveraging homodyne measurements and rotations, we construct a collection of
locally optimal POVMs, which allows us to apply the mathematical framework of
AQSE to Gaussian measurements and feedback \cite{Nagaoka2005,Fujiwara2006}.
Through formal mathematical analysis, we show that this adaptive measurement
process allows for extracting the maximum possible information pertaining to the
phase encoded in squeezed vacuum states in the asymptotic regime.
Appendix~\ref{Appendix_Proof} gives a mathematical proof of the asymptotic
optimality of the adaptive strategy, showing its convergence to the quantum
Cramér-Rao lower bound (QCRB). In Sec.~\ref{sec:Performance}, we use numerical
simulations to evaluate the performance of this strategy, and investigate its
performance under losses and system imperfections in Sec.~\ref{sec:Losses}. We
observe that this strategy approaches the quantum limit for phases within
$[0,\pi/2)$, outperforming previous phase estimation strategies.
Sec.~\ref{sec:General-dyne}, describes the combined homodyne-heterodyne strategy
for phase estimation in the full range for squeezed vacuum of $[0,\pi)$, and the
proof of its asymptotic optimality in Appendix \ref{AppendixExtension}. Finally,
Sec.~\ref{sec:Diss} contains the discussion and concluding remarks.

\section{\label{sec:II} Background}

\subsection{\label{sec:QPE} Single parameter estimation in quantum systems}

A fundamental problem in quantum parameter estimation is the design of precise
estimators of an unknown parameter $\theta \in \Theta$ characterizing a quantum
state based on measurements of the system. In this context, a quantum system is
modeled as a Hilbert space $\mathcal{H}$, and its state is described by a
density operator $\rho$, which is a self-adjoint positive operator with unit
trace on $\mathcal{H}$. The process of encoding the unknown parameters into a
probe state $\rho$ is accomplished by a dynamical process, which, when it can be
represented as a unitary transformation $U(\theta)$, yields the state
\begin{equation}
  \label{eq:state_par}
\rho(\theta) = U(\theta)\rho U^{\dagger}(\theta), \quad \theta \in \Theta.
\end{equation}

Estimation of $\theta$ can be achieved through an estimator that is a function
that takes a sample of size $N$ from a measurement of the quantum system as an
input and produces an estimate of the unknown parameter. The most general
description of a measurement process is a POVM \cite{Holevo2011}. Given a
quantum system $\mathcal{H}$ and an outcome space $\mathcal{X} \subseteq
\mathbb{R}^{k}$ for a measurement, a POVM is a map $M:\mathcal{B}(\mathcal{X})
\to B(\mathcal{H})$ from the set of events of our random experiment
$\mathcal{B}(\mathcal{X})$ to the space of bounded operators on $\mathcal{H}$,
denoted by $B(\mathcal{H})$, that satisfies the following conditions
\cite{Holevo2011,Chiribella2006,Beneduci2011}:
  \label{eq:POVM}
  \begin{enumerate}
\item[i.] $M(\emptyset) = 0, \quad M(\mathcal{X}) = I$
\item[ii.] $M(B) \geq 0, \quad \forall B \in \mathcal{B}(\mathcal{X})$
\item[iii.] For every family of mutually disjoint events $\left\{ B_n
  \right\}_{n=1}^{\infty} \subset \mathcal{B}(\mathcal{X})$, so that $B_i \cap
  B_j = \emptyset$ $\, \forall i \neq j$, that satisfies
  $\cup_{j=1}^{\infty}B_j = B \in \mathcal{B}(\mathcal{X})$, then $M(B) =
  \sum_{j = 1}^{\infty} M(B_j)$.
\end{enumerate}
In particular, when the state of the system is $\rho(\theta)$, the observed data
$x \in \mathcal{X}$ of a measurement $M$ is an outcome of a random variable $X
\in \mathcal{X}$ distributed according to the density function $f(x \mid \theta;
M)$ (or probability mass function in the case of discrete random variables)
given by the Born’s rule
\begin{equation}
  \label{eq:Born_rule}
  f(x \mid \theta; M) = \textrm{Tr}\left[ M(x) \rho(\theta) \right].
\end{equation}

Thus, any sample from the application of a sequence of $N$ POVMs $M_1,\ldots,
M_N$ in a quantum system is represented as a sequence of $N$ random variables
$\vec{X}_N= X_1,\ldots,X_N $. It follows that any estimator
$\widehat{\theta}\left(X_1,\ldots,X_N\right)$ based on this sample is also a
random variable with expected value
\begin{widetext}
\begin{equation}
  \label{eq:expected_value}
    \textrm{E}_{\theta}\left[ \widehat{\theta}(\vec{X}_N) \right] = \int_{\mathcal{X}^{N}} \widehat{\theta}(x_1,\ldots,x_N) f(x_1,\ldots,x_N \mid \theta; M_1,\ldots,M_N) dx_1 \cdots dx_N.
\end{equation}
\end{widetext}
To find optimal estimators for all $\theta \in \Theta \subset \mathbb{R}$, the
concept of unbiased estimator plays a crucial role. An estimator
$\widehat{\theta}(\vec{X}_N)$ is unbiased if $\mathrm{E}_\theta\left[
  \widehat{\theta}(\vec{X}_N) \right] = \theta$ for all $\theta \in \Theta$
\cite{Keener2010}. The performance of unbiased estimators is characterized by
their variance, which is bounded by the Cramér-Rao bound \cite{Keener2010}. This
bound corresponds to the classical limit of precision for all unbiased
estimators, and is given by the inverse of the Fisher information denoted by
$F_X(\theta)$. Given a sample $X$ produced by a POVM $M$, the Fisher information
\begin{equation}
  \label{eq:Fisher_inf}
  F_X(\theta) = \int_{\mathcal{X}} f(x \mid \theta; M) \left[ \frac{\partial  }{ \partial \theta} \log\left( f(x \mid \theta; M) \right) \right]^2 dx
\end{equation}
quantifies the amount of information about the parameter $\theta$ that can be
extracted from the sample $X$ \cite{Braunstein1994,Nagaoka2005}. The ultimate
limit of precision, dictated by quantum mechanics, is achieved by optimizing
$F_X(\theta)$ over all possible POVMs, resulting in the quantum Fisher
information (QFI). Consequently, the variance of any unbiased estimator is lower
bounded by the inverse of the QFI, referred to as the quantum Cramér-Rao bound
(QCRB) \cite{helstrom1969quantum, Braunstein1994,Nagaoka2005}. The objective in
quantum parameter estimation is to devise estimators and quantum measurement
schemes that attain the QCRB for any value of the parameter $\theta \in \Theta$.

\subsection{\label{sec:Dyne} Optical phase estimation based on dyne-detection}

A central task in optical quantum metrology is the estimation of an unknown
phase $\theta \in [0, 2\pi)$ encoded in a photonic quantum state by the unitary
process $U(\theta) = e^{-i \hat{n} \theta}$, where $\hat{n}$ is the photon
number operator. The standard quantum state probe for optical phase estimation
is the coherent state $\ket{\alpha}\bra{\alpha}, \, \alpha \in \mathbb{C}$, in
which photons exhibit classical correlations \cite{PhysRev.131.2766,
  polino2020photonic}. For this quantum state, the QCRB for any unbiased
estimator $\widehat{\theta}$ and $N$ independent copies of the system is
\cite{Holevo2011}
\begin{equation}
  \label{eq:coh_states}
  \mathrm{Var}[\widehat{\theta}] \geq \frac{1}{4 N  \mathrm{E}\left[ \hat{n} \right]},
\end{equation}
This limit in precision defines the SNL (or the coherent state limit).

To surpass the SNL for optical phase estimation, it is necessary to employ
states with quantum correlations, such as squeezed vacuum states. These states
are defined by the density operator \cite{polino2020photonic}:
\begin{equation}
\label{eq:squeezed}
\rho_r = \ket{0,r}\bra{0,r} = \hat{S}(r) \ket{0}\bra{0} \hat{S}^{\dagger}(r),
\end{equation}
where $\hat{S}(r) = e^{ \frac{1}{2}\left( re^{-\iu \gamma}\hat{a}^2 - re^{\iu \gamma}
  \hat{a}^{\dagger} \right) }$ is the squeezing operator, $r \in \mathbb{R}$,
$\gamma \in [0, 2\pi)$, and $\hat{a}$ and $\hat{a}^{\dagger }$ are the
annihilation and creation bosonic operators, respectively. Through the unitary
transformation $U(\theta)$, the squeezed vacuum state in Eq.~\eqref{eq:squeezed}
results in the parameter-dependent state $\rho(\theta)=U^{\dagger}(\theta)\rho
U(\theta)$. The corresponding QCRB for this state for any unbiased estimator
$\widehat{\theta}$ and $N$ independent copies of the system is given by
\cite{Monras2006}:
\begin{equation}
\label{eq:QCRB_sq}
\mathrm{Var}[\widehat{\theta}] \geq \frac{1}{2 N \sinh^2(2r)} = \frac{1}{8 N \left( \mathrm{E}\left[ \hat{n} \right]^2 + \mathrm{E}\left[ \hat{n} \right] \right)},
\end{equation}
where $r$ represents the squeezing strength of $\hat{S}(r)$
\cite{PhysRevA.79.033834,PhysRevResearch.3.023222}. This bound exhibits a
superior scaling with $\mathrm{E}\left[ \hat{n} \right]$ compared to the
coherent state in Eq.~\eqref{eq:coh_states}. However, it is worth noting that
due to the $\pi$-inversion symmetry inherent to squeezed vacuum states
\cite{zeytinouglu2017engineering} (see Fig.~\ref{fig:scheme}-\textbf{i}), any
estimation strategy based on these states is constrained to phases within the
range of $[0, \pi)$.

\begin{figure*}[h!t]
  \centering
  \includegraphics[scale = 0.56]{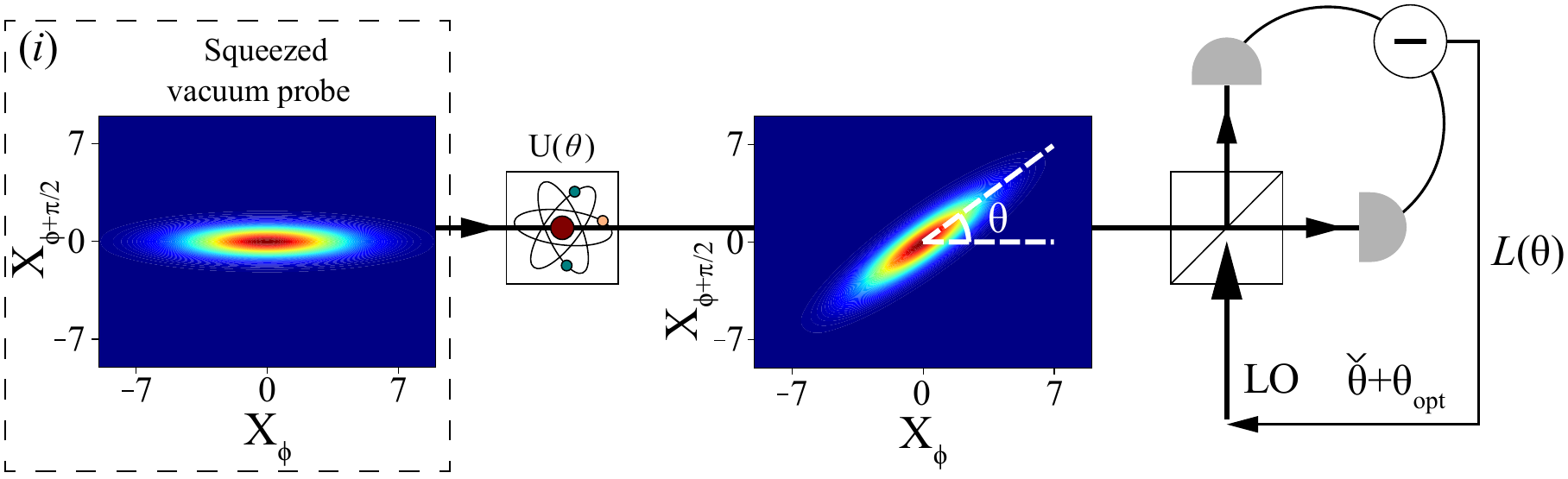}
  \caption{\label{fig:scheme} Adaptive estimation strategy for optical phase
    estimation with squeezed vacuum probe states $\lvert 0, r\rangle \langle 0,
    r \rvert$. The strategy employs locally optimal POVMs, $M_{\check{\theta}}$
    in Eq. (\ref{eq:set_of_lues}) with $\check{\theta} \in [0, \pi/2)$, to
    produce a maximum likelihood estimate of $\theta$, and updates the value
    $\check{\theta}$ for subsequent adaptive steps. The measurement process is
    iteratively repeated during the adaptive strategy. Inset
    \textbf{(\textbf{$i$})} shows the Husimi \textit{Q} representation for the
    initial squeezed vacuum state. Note that due to the internet symmetry
    properties of squeezed vacuum, these quantum probes can only encode the
    phase modulo $\pi$.}
\end{figure*}

The standard measurement approach for optical phase estimation is the heterodyne
measurement \cite{RevModPhys.84.621, Tyc2004}. This measurement involves
simultaneous sampling of two orthogonal components of the electromagnetic field
within the complex plane, namely $\hat{X}_{\phi}$ and $\hat{X}_{\phi+\pi/2}$, by
utilizing the quadrature decomposition of the input field
\cite{Wiseman2009_book}. Here, the quadrature operator $\hat{X}_{\phi}$ is
defined as:
\begin{equation}
  \label{eq:Quadrature_operator}
  \hat{X}_{\phi} = \frac{\hat{a}^{\dagger} e^{i\phi} + \hat{a} e^{-i\phi}}{2}.
\end{equation}
The POVM associated with the heterodyne measurement is described by coherent
state projectors $M_{\rm{Het}} = \left\{ \pi^{-1} \ket{z}\bra{z} : z \in
  \mathbb{C} \right\}$, with outcomes corresponding to complex numbers, and with
the corresponding Fisher information:
\begin{equation}
  \label{eq:FI_heterodyne}
  F_Z(\theta) = 4\sinh^{2}\left( r \right).
\end{equation}
The inverse of Eq.~\eqref{eq:FI_heterodyne} is known as the heterodyne limit for
the precision of any unbiased estimator $\widehat{\theta}$ for all $\theta \in
\left[ 0, \pi \right)$ with squeezed vacuum states.

Going beyond estimation strategies based on heterodyne detection, homodyne
detection can surpass the heterodyne limit for a suitable set of values of
$\theta$. Homodyne provides information about the quadrature $\hat{X}_{\phi}$ of
the input signal using a local oscillator (LO) phase reference field and
interference \cite{Wiseman2009_book}. Specifically, in the limit of strong LO,
Eq.~\eqref{eq:Quadrature_operator} represents a self-adjoint operator with a
spectral measure given by \cite{RevModPhys.84.621}:
\begin{equation}
  \label{eq:Spectral}
  \hat{X}_{\phi} = \int_{-\infty}^{\infty} x \Pi(dx),
\end{equation}
where $\Pi(B) = 1_{B}(x)$ (or symbolically in Dirac notation, $\Pi(dx) =
\ket{x}\bra{x} dx$), for any $B \in \mathcal{B}\left( \mathbb{R} \right)$.
Consequently, the homodyne measurement can be described by the POVM
$M_{\rm{Hom}} = \left\{ \Pi(dx) \right\}$, with the outcome space being the real
numbers \cite{d2003quantum}.

For squeezed vacuum state probes in Eq.~\eqref{eq:squeezed}, the outcomes $x \in \mathbb{R}$ of the homodyne measurement are distributed according to a normal random variable with probability density function
\begin{equation}
  \label{eq:pdf_homodyne}
  f(x \mid \theta) = \frac{1}{\sqrt{2 \pi \sigma^2(\theta)}}\exp\left[ -\frac{x^2}{2 \sigma^2(\theta)} \right],
\end{equation}
where $\theta$ denotes the unknown phase and
\begin{equation}
  \label{eq:var_homodyne}
  \sigma^2(\theta) = \left[ e^{-2r} \cos^2(\theta) + e^{2r}\sin^2(\theta)
  \right]
\end{equation}
denotes the variance. The Fisher information $F_X(\theta)$ for the
homodyne measurement is then
\begin{equation}
  \label{eq:Fisher_Homodyne}
  \begin{split}
    F_X(\theta) = \frac{ 2 \sinh^2(2r) \sin^2(2 \theta) }{ \left(  \sigma^2(\theta)\right)^2 }.
  \end{split}
\end{equation}
Notably, the classical Fisher information of the homodyne measurement coincides
with the QFI when the squeezing strength $r$ satisfies:
\begin{equation}
  \label{eq:opt_squeez}
  r = -\frac{1}{2}\log(\tan(\theta)),
\end{equation}
or equivalently, when the parameter $\theta$ corresponds to the optimal value
$\theta_{\mathrm{opt}}$ given by
\begin{equation}
  \label{eq:theta_opt}
  \theta_{\mathrm{opt}} = \frac{\arccos\left( \tanh(2r) \right)}{2}
\end{equation}
which tends to zero as $r$ increases. Consequently, the homodyne measurement can
surpass the heterodyne limit in the neighborhood of $\theta_{\mathrm{opt}}$.
However, outside this neighborhood, estimators based on a sample $\vec{X}_N$
obtained from $N$ independent and identical homodyne measurements cannot achieve
this optimal level of precision (see Fig. \ref{fig:performance_2}).

To overcome this limitation, adaptive estimation protocols have been proposed.
One such protocol \cite{Monras2006, Olivares2009, Berni2015}, that we refer to
as the two-step protocol, considers a reduced parameter space $[0, \pi/2)$, and
utilizes homodyne detection and one subsequent adaptation of the probe state to
surpass the heterodyne limit within this phase range. This strategy approaches
the QCRB in the asymptotic limit for phases around the optimal phase
$\theta_{\mathrm{opt}}$ \cite{Olivares2009, Berni2015}. However, far from this
optimal phase, its performance significantly deviates from the QCRB. Moreover,
when considering the full range of phases that can be encoded in squeezed vacuum
probes $[0,\pi)$, this two-step estimation strategy is not expected to produce
satisfactory results, due to the periodicity of the likelihood function from
homodyne outcomes.

In a more general measurement setting, Gaussian \cite{RevModPhys.84.621,Oh2019}
and generalized dyne measurements \cite{Genoni2014, PhysRevResearch.2.023030},
which extend the concepts of homodyne and heterodyne, have a large potential for
quantum metrology \cite{Oh2019}, sensing \cite{Lee2021}, and communications
\cite{Holevo2021}, and for studying the dynamical evolution of quantum systems
under continuous measurements \cite{Genoni2014, PhysRevLett.95.030402}.
Moreover, the combination of optimal control, quantum feedback, and Gaussian
measurements allows for implementations of optimal phase measurements for single
qubits, even across the full range of phases \cite{Wiseman1996, Martin2020}.
However, optimal phase measurements based on quantum feedback and optimal
control for quantum correlated states, such as squeezed states, are still
unknown, but are expected to be highly complex in practice.

In this work, we propose an adaptive estimation strategy based on homodyne
detection, that leverages the framework of AQSE, which, under certain regularity
conditions (see Appendix:~\ref{sec:MLE}), yields a consistent and efficient
estimator for any phase $\theta \in [0,\pi/2)$. A key element of this approach
is to use samples that lead concave likelihood functions, ensuring the asymptotic
normality of the Maximum Likelihood Estimator (MLE) \cite{Nagaoka2005,
  Fujiwara2006}. This property guarantees the asymptotic saturation of the QCRB
in Eq.~\eqref{eq:QCRB_sq} for any $\theta \in \left[ 0, \pi/2 \right)$. We
further generalize this adaptive strategy to incorporate heterodyne measurements
enabling phase estimation within $[0,\pi)$, while maintaining a quantum-optimal
performance in the asymptotic limit.

\section{\label{sec:Op_dyne} Optimal adaptive homodyne phase estimation}

In practice, it is generally impossible to find a POVM and an unbiased estimator
capable of saturating the QCRB for all $\theta \in \Theta$. However, it is often
possible to find a POVM and an unbiased estimator that can achieve this bound
for a specific value of the parameter within a neighborhood around a point
$\theta_0 \in \Theta$. These types of POVMs are referred to as locally optimal
at $\theta_0$. Moreover, if it is possible to construct a collection of such
locally optimal POVMs for any $\theta \in \Theta$ while satisfying a set of
regularity conditions pertaining to the probability distributions of their
outcomes (see Appendix:~\ref{sec:MLE}), then it is possible to use an adaptive
estimation method, known as AQSE \cite{Fujiwara2006}, capable of saturating the
QCRB in the asymptotic limit for the MLE $\forall \theta \in \Theta$.

Building upon this understanding, the proposed adaptive phase estimation
strategy with squeezed vacuum states is constructed based on two elements. The
first element involves the construction of a set of POVMs through homodyne
measurements that are locally optimal for any value of the parameter, that is
the optical phase $\theta$ within the range $\theta \in [0, \pi/2) = \Theta$.
The second element is the construction of an estimator that achieves the QCRB
for any given value of the phase $\theta$, while also being locally unbiased at
the true phase $\theta_0 \in \Theta$ \cite{Holevo2011, Fujiwara2006}.

To construct the set of locally optimal POVMs, we refer to
Eq.~\eqref{eq:theta_opt}, which shows that the homodyne POVM $M_{\mathrm{Hom}} =
\left\{ \Pi(dx) \right\}_{x \in \mathbb{R}}$ in Eq.~\eqref{eq:Spectral} is
locally optimal at $\theta_{\mathrm{opt}}$. This is because the samples obtained
from $M_{\mathrm{Hom}}$ have a Fisher information equal to the QFI at this
specific phase. Furthermore, given the asymptotic unbiasedness of the MLE, and
its subsequent local unbiasedness, the estimator effectively saturates the QCRB
(Eq.~\eqref{eq:QCRB_sq}) in the asymptotic limit at $\theta_{\mathrm{opt}}$.
Consequently, by appropriately incorporating a phase shift into the elements of
$M_{\mathrm{Hom}}$, it is possible to construct a set of locally optimal POVMs
for any phase $\theta \in [0, \pi/2)$.

To this end, we introduce a phase shift $U\left(\check{\theta} -
  \theta_{\mathrm{opt}}\right)$, which allows us to define a new set of POVMs
for each $\check{\theta} \in [0, \pi/2)$ as follows:
\begin{equation}
  \label{eq:set_of_lues}
	M_{\check{\theta}}(dx) = \left\{  U\left( \check{\theta} - \theta_{\mathrm{opt}} \right)\Pi(dx) U^{\dagger}\left( \check{\theta}- \theta_{\mathrm{opt}} \right) \right\}.
\end{equation}

Here $M_{\check{\theta}}(dx)$ denotes the POVM elements obtained by applying the
phase shift to the original POVM elements $\Pi(dx)$ of $M_{\mathrm{Hom}}$. Note
that the phase distribution of $M_{\check{\theta}}(dx)$ over the state
$\rho(\theta)$ becomes:
\begin{align}
  \label{eq:lue_homodyne}
 \resizebox{0.22 \hsize}{!}{$ f(x \mid \theta; M_{\check{\theta}} )$} &\resizebox{0.78 \hsize}{!}{$
  	= \textrm{Tr}\left[ U\left( \check{\theta} - \theta_{\mathrm{opt}} \right)\ket{x}\bra{x} U^{\dagger}\left( \check{\theta}- \theta_{\mathrm{opt}} \right) \rho(\theta)   \right] \nonumber  $}\\
                                          &= f(x \mid \theta  + \theta_{\mathrm{opt}} -\check{\theta}; M_{\mathrm{Hom}} ).
\end{align}
Evaluated at $\check{\theta} = \theta$, this distribution is the same as the
distribution for the outcomes of the POVM $M_{\mathrm{Hom}}$ at
$\theta_{\mathrm{opt}}$, which shows that the POVM $M_{\check{\theta}}(dx)$ is
locally optimal at $\theta$. From this observation and based on the AQSE
framework, then it is possible to construct a multi-step estimation strategy
based on adaptive homodyne measurements, which provide a set of locally optimal
measurements, for which the distribution of the sequence of estimators converges
to a normal distribution with variance equal to the inverse of the QFI (see
Appendix~\ref{Appendix_Proof}).

Figure~\ref{fig:scheme} shows the concept of the proposed phase estimation
strategy based on adaptive homodyne measurements. For a set of $N$ input
squeezed probe states $\{\ket{0,r}\bra{0,r}\}$, a phaseshift $U(\theta)$ encodes
the parameter $\theta$ in the probes. The adaptive strategy then implements the
POVM $M_{\check{\theta}_1}$ with an initial (random) guess $\check{\theta}_1 \in
[0, \pi/2)$ over $\nu = N / m$ of these states, yielding a measurement sample
$\vec{X}_{\nu}(\check{\theta}_1) = X_1(\check{\theta}_1), \ldots,
X_\nu(\check{\theta}_1)$. The MLE applied to $\vec{X}_{\nu}(\check{\theta}_1)$,
$\widehat{\theta}_{\mathrm{MLE}}\left( \vec{X}_{\nu}(\check{\theta}_1) \right)$,
results in an estimate $\check{\theta}_2=\widehat{\theta}_{\mathrm{MLE}}\left(
  \vec{x}_{\nu}(\check{\theta}_1), \right)$ of $\theta$ for the firt adaptive
step. This estimate $\check{\theta}_2$ then becomes the best guess for the
subsequent adaptive step, and the process is repeated $m$ times iteratively
during the strategy.

\subsection{\label{sec:Selection_of_MLE} Estimator}

\begin{figure*}[h!t]
  \centering \includegraphics[scale = 0.295]{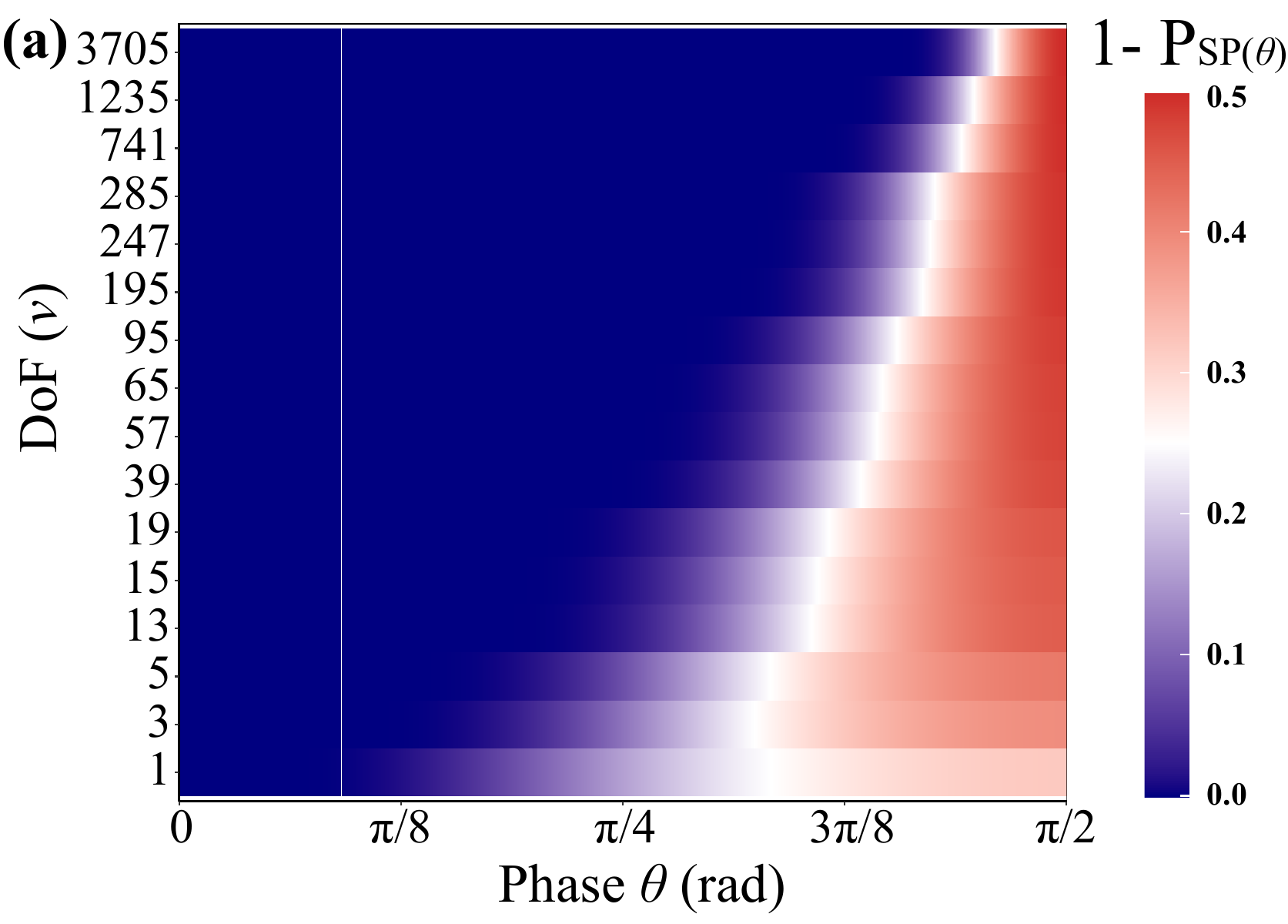}
  \includegraphics[scale = 0.295]{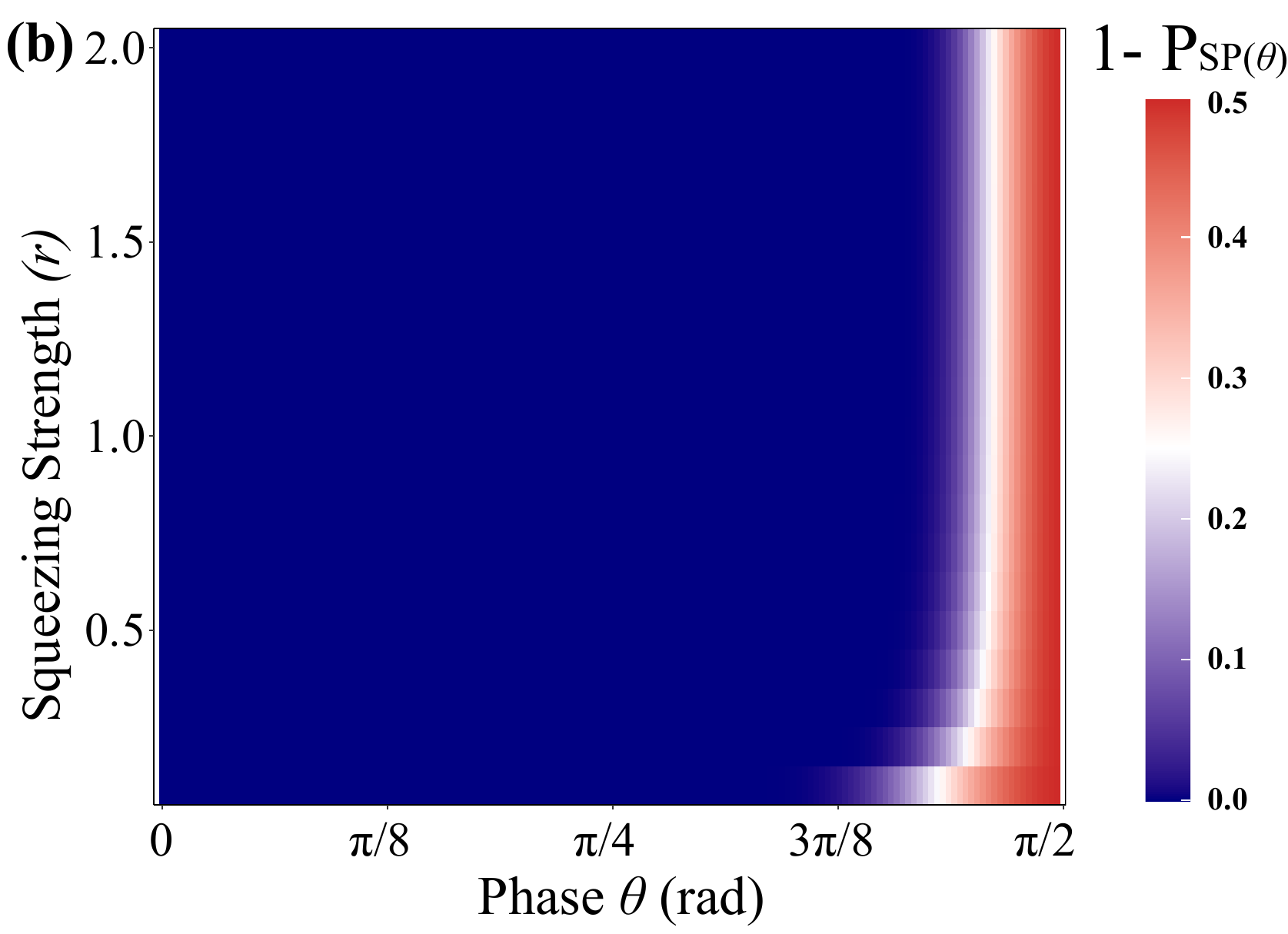}
  \caption{\label{fig:probas} Probability that the likelihood function does not
    have its global maximum at a stationary point $1-P_{SP}(\theta)$ within
    $\left[ 0,\pi/2 \right)$ from Eq.~\eqref{eq:prob_stationary_point}.
    \textbf{(a)} Probability $1-P_{SP}(\theta)$ for different degrees of freedom
    (DoF) $\nu$ (number of probes in each adaptive step) with $r = 1$.
    \textbf{(b)} Probability $1-P_{SP}(\theta)$ for different values of
    squeezing strength $r$ with $\nu = 3705$. By ensuring that the global
    maximum of the likelihood function is attained at a stationary point within
    the interval $[0, \pi/2)$, we mitigate the bias introduced by boundary
    estimates and enhance the performance of the adaptive estimation process
    (See main text for details).}
\end{figure*}

A key aspect of the parameter estimation strategy is the selection of an
estimator that allows for the saturation of the QCRB. We identify the necessary
conditions to ensure the asymptotic consistency and normality of the MLE and the
saturation of the QCRB through the adaptive strategy for any $\theta \in \left[
  0, \pi/2 \right)$. Assuming that the Fisher information is different from zero
for any $\theta$ \footnote{This assumption holds except in the first adaptive
  step if $\check{\theta}_1 = \theta + \theta_{\text{opt}}$. However, in our
  adaptive strategy, $\check{\theta}_{1}$ is chosen at random within
  $[0,\pi/2)$. Therefore the probability of observing this event is zero.}, the
MLE saturates the QCRB given that:

\begin{itemize}
\item[(a)] the MLE is a single-valued function;
\item[(b)] the MLE is obtained as a stationary point of the likelihood function;
  and
\item[(c)] the derivatives of the likelihood function at $\theta$ exist up to a
  high-enough order so that they can be effectively approximated using a Taylor
  series \cite{Makelainen1981, Fujiwara2006}.
\end{itemize}

We observe that (a) is automatically satisfied, since the measurement outcomes
from the proposed strategy follow a normal distribution, as shown
Eq.~(\ref{eq:pdf_homodyne}). This condition ensures the sufficient smoothness of
the likelihood function for any $\theta \in [0, \pi/2)$ in each adaptive step.
Moreover, by restricting the parameter space to the interval $\left[0,
  \pi/2\right)$, this condition guarantees the uniqueness of the MLE for any
$\theta \in [0, \pi/2)$. Thus, the remaining task is to determine the conditions
under which the MLE corresponds to a stationary point of the likelihood
function, leading to its asymptotic normality. To this end, we first analyze the
MLE for a sample from the first adaptive step.

Given a sample $\vec{X}_{\nu}(\check{\theta}_1) =\left( X_1,X_2,\ldots,X_\nu
\right)$ of size $\nu$ from the POVM $M_{\check{\theta}_1}$ at the first
adaptive step, and evaluating the likelihood from Eq.~\eqref{eq:pdf_homodyne} at
$\theta_{*} = \left( \theta + \theta_{\mathrm{opt}} -\check{\theta}_1 \right)
\,$ $\mathrm{modulo} \, $ $\pi/2$, we obtain the MLE for $\theta$ as:
\begin{equation}
  \label{eq:mle_1}
  \resizebox{.98\hsize}{!}{$\widehat{\theta}_{\mathrm{MLE}}( \vec{X}_{\nu} ) = \arccos\left[  \frac{e^r \sqrt{e^{2r} - \frac{1}{\nu}\sum_{i=1}^{\nu} X_i^2   }}{\sqrt{e^{4r}-1}} \right] - \theta_{\mathrm{opt}} + \check{\theta}_1.$}
\end{equation}
The set of homodyne outcomes for which $e^{2r} - \frac{1}{\nu}\sum_{i=1}^{\nu} X_i^2 > 0$ in  Eq.~(\ref{eq:mle_1}), which results in a real-valued $\widehat{\theta}_{\mathrm{MLE}}( \vec{X}_{\nu} )$, yields a MLE corresponding to a stationary point of the likelihood within $\left[ 0,\pi/2 \right)$. Therefore, the probability of obtaining a non-real solution
of Eq.~\eqref{eq:mle_1} at $\check{\theta}_1 = \theta_{\mathrm{opt}}$ is:
\begin{widetext}
\begin{equation}
  \label{eq:prob_stationary_point}
  \begin{split}
    P\left( e^{2r} < \frac{1}{\nu}\sum_{i=1}^{\nu} X_i^2 \mid \check{\theta}_1 = \theta_{\mathrm{opt}}   \right)  &= P\left( \sum_{i=1}^{\nu} \frac{X_i^2}{\sigma^2(\theta)} > e^{2r}\left( \frac{\nu}{\sigma^2(\theta)} \right)    \right)\\
                                                                                                                    &= 1- P\left( \sum_{i=1}^{\nu} \frac{X_i^2}{\sigma^2(\theta)} \leq e^{2r}\left( \frac{\nu}{\sigma^2(\theta)} \right) \right)\\
                                                                                                                    &=1-P_{SP}(\theta).
    \end{split}
  \end{equation}
\end{widetext}
Here, $P_{SP}(\theta)$ is the probability that the likelihood function has its global maximum at
a stationary point at $\theta$. We note that $Q = \sum_{i=1}^{\nu} \frac{X_i^2}{\sigma^2(\theta)}$ is the sum of squares of $\nu$ independent standard normal random variables and follows a chi-squared
distribution with $\nu$ degrees of freedom ($Q \sim \chi^2(\nu)$)
\cite{Keener2010}.

Figure~\ref{fig:probas} shows the probability of obtaining a non-real solution
of Eq.~\eqref{eq:mle_1} as a function of $\theta$ for different values of the
squeezing strength $r$ and sample size $\nu$. We observe that as $\theta$
deviates from the optimal value $\theta_{\mathrm{opt}}$ in
Eq.~\eqref{eq:theta_opt}, the probability that the estimates were not obtained
from a stationary point in the likelihood function within the interval $\left[
  0,\pi/2 \right)$ becomes different than zero. We also note that the region in
which the global maximum of the likelihood is not reached at a stationary point
decreases as we increase the degrees of freedom $\nu$ (sample size in the
adaptive step) or the squeezing strength $r$. Moreover, when $e^{2r} - 1/ \nu
\sum_{i=1}^{\nu} X_i^2 < 0$ (regions where the global maximum does not
correspond to a stationary point) the MLE in Eq.~(\ref{eq:mle_1}) corresponds to
the boundary point $\pi/2$. This event introduces a bias in the estimate for
subsequent adaptive steps, leading to a decrease in the precision of the final
estimate. This effect in the estimate becomes more detrimental when the phase to
be estimated $\theta_{0}$ is close to the boundary point $\pi/2$ (see
Figure~\ref{fig:probas}) as can be observed in the previous two-steps protocols
\cite{Berni2015}. We can attribute this issue to the Fisher information being
zero at $\theta = \pi/2$ (or $\theta = 0$). As a result, the information that
can be obtained about the phase at these boundary points vanishes.

The proposed strategy address this problem by first making a random guess
$\check{\theta}_1\in[0,\pi/2)$ in the first adaptive step. This initial random
guess reduces in average the probability that the MLE does not arise from a
stationary point within $[0,\pi/2)$ in Eq.~(\ref{eq:prob_stationary_point}).
According to the law of total probability
\begin{widetext}
\begin{equation}
  \label{eq:prob_ht}
  \begin{split}
    P\left( e^{2r} < \frac{1}{\nu}\sum_{i=1}^{\nu} X_i^2   \right) &=  \mathrm{E}_{\hat{\theta}_0}\left[   P\left( e^{2r} < \frac{1}{\nu}\sum_{i=1}^{\nu} X_i^2  \mid \check{\theta}_1  \right) \right]\\
    &= \int_{0}^{\pi/2} d\check{\theta}_1 P\left( \sum_{i=1}^{\nu} \frac{X_i^2}{\sigma^2\left(  \theta_{*}  \right)} > e^{2r}\left( \frac{\nu}{\sigma^2\left(\theta_{*} \right)} \right)    \right)\\
    &= \int_{0}^{\pi/2} d{\check\theta}_1 P_{SP}(\theta_{*}).
\end{split}
\end{equation}
\end{widetext}
Moreover, this probability decreases as the sample size $\nu$ increases, as
shown in Fig.~\ref{fig:probas_2}. As a final step, to guarantee that the MLE
always arises from stationary points, we introduce a modified estimator
\begin{equation}
  \label{eq:estimator_proposal}
  \widehat{\theta}^{\textrm{U}}_{\textrm{MLE}}(\vec{X}_{\nu}) = \begin{cases} \widehat{\theta}_{\textrm{MLE}}(\vec{X}_{\nu}) &\text{if } e^{2r} - 1/ \nu \sum_{i=1}^{\nu} X_i^2 \geq 0, \\ \widehat{\theta}_{\textrm{MLE}}(\vec{X}_{s})  &\text{otherwise}, \end{cases}
\end{equation}
where $\vec{X}_{s}$ is a subsequence of $\vec{X}_{\nu}$ constructed by
iteratively removing the highest values of $\vec{X}_\nu$ in descending order
until the condition $e^{2r} - 1/ s \sum_{i=1}^{s} X_i^2 \geq 0$ is satisfied. We
note that while the estimator ignores a few measurement outcomes from the
sample, the probability of this event happening is low, and this modified
estimator greatly improves the final variance of the estimates. Moreover, it is
worth noting that the modified estimator in Eq.~(\ref{eq:estimator_proposal}) is
only necessary in the first adaptive step. Once an estimate $\check{\theta}$
sufficiently close to the true value $\theta$ is obtained, the subsequent
adaptive steps will result in samples close to $\theta_{\mathrm{opt}}$.
Therefore, for sufficiently large values of $r$, and moderate $\nu$, this
procedure makes the probability in Eq.~\eqref{eq:prob_stationary_point} tend to
zero, guaranteeing the asymptotic efficiency of the MLE from the adaptive
strategy.
  \begin{figure}[h!t]
  \centering
  \includegraphics[scale = 0.33]{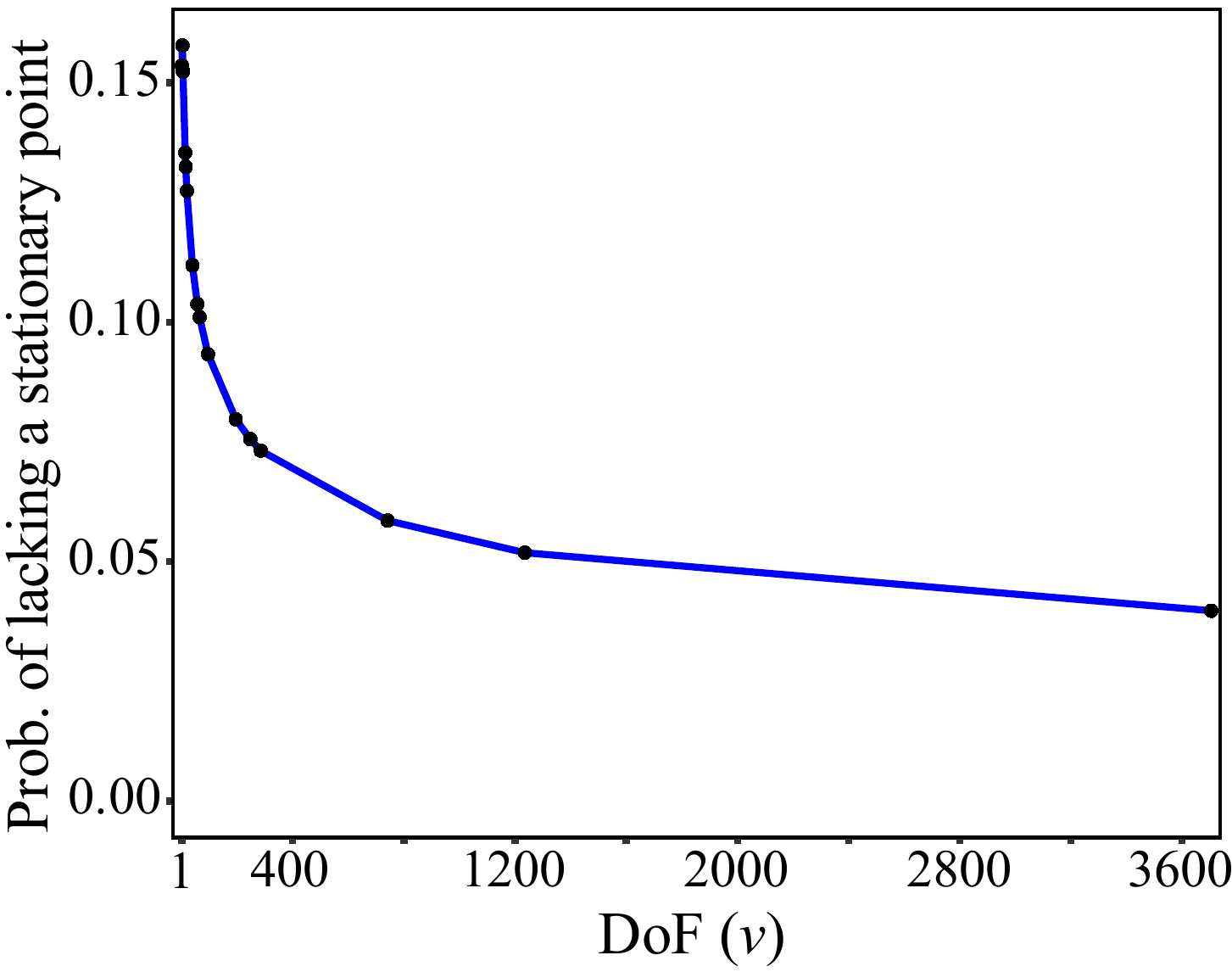}
  \caption{\label{fig:probas_2} Conditional probability (Eq.~\eqref{eq:prob_ht})
    that the likelihood function does not have its global maximum at a
    stationary point within $\left[ 0,\pi/2 \right)$, given $\check{\theta}_1$
    uniformly distributed from $0$ to $\pi/2$. The black points show numerical evaluations of (Eq.~\eqref{eq:prob_ht}) for several $\nu$, and the blue line is a guide for the eye showing the trend as a function of $\nu$. This probability is upper bounded
    for $\nu = 1$, and decreases as $\nu$ increases. The parameter $\nu$
    corresponds to the degrees of freedom (DoF) for $Q = \sum_{i=1}^{\nu}
    \frac{X_i^2}{\sigma^2(\theta)}$, which follows a chi-square distribution.}
\end{figure}

\subsection{\label{sec:Selection_of_MLE} Proof of optimality of the adaptive estimation strategy}

Appendix~\ref{Appendix_Proof} gives a mathematical proof of the optimality of the
proposed estimation strategy based on adaptive homodyne detection. The proof
first establishes that by satisfying the regularity conditions outlined in
Appendix \ref{sec:MLE}, the probability of obtaining an estimate far from the
true value $\theta_0$ decreases exponentially to zero as the number of adaptive
steps increases. This result demonstrates the \textit{almost sure convergence}
of the MLE to $\theta_0$ (asymptotic consistency) \cite{Keener2010}.
Subsequently, using the asymptotic consistency of the estimator, we prove that
the limiting distribution of the MLE, as the number of adaptive steps tends to
infinity, is a normal distribution with mean $\theta_0$ and variance equal to
the inverse of the QFI (asymptotic normality). This demonstrates the convergence
of the proposed adaptive strategy to the QCRB.

\section{\label{sec:Performance}Performance of the adaptive strategy}

We evaluate the performance of the adaptive estimation strategy based on locally
optimal POVMs, Eq.~(\ref{eq:set_of_lues}), with the modified estimator in
Eq.~(\ref{eq:estimator_proposal}). We conduct Monte Carlo simulations varying
the sample size $\nu$, number of adaptive steps $m$, and squeezing strength $r$.
We investigate the precision and efficiency of the estimation strategy and its
convergence towards the QCRB in Eq.~(\ref{eq:QCRB_sq}) for $\theta \in [0,
\pi/2)$. Considering that the estimator
$\widehat{\theta}\left(\vec{X}_\nu(\check{\theta}_1)
  ,\ldots,\vec{X}_\nu(\check{\theta}_m)\right)$ has a periodic distribution with
period $\pi/2$, we evaluate the precision of
$\widehat{\theta}\left(\vec{X}_\nu(\check{\theta}_1)
  ,\ldots,\vec{X}_\nu(\check{\theta}_m)\right)$ by using the corresponding
Holevo variance \cite{Holevo2011, Berry2000}:
\begin{widetext}
\begin{equation}
  \label{eq:holevo_var}
  \mathrm{Var}_{\theta}\left[ \widehat{\theta}\left(\vec{X}_\nu(\check{\theta}_1)
      ,\ldots,\vec{X}_\nu(\check{\theta}_m)\right) \right] = \frac{\left[\mathrm{E}\left[\cos\left(\frac{2\pi}{P}\left( \widehat{\theta}\left( \vec{X}_\nu(\check{\theta}_1)
              ,\ldots, \vec{X}_\nu(\check{\theta}_m)\right) -\theta\right)\right)\right]\right]^{-2}-1}{\left(\frac{2\pi}{P}\right)^2},
\end{equation}
\end{widetext}
where the factor $P$ represents the period of the estimator's distribution, in
our case $P = \pi/2$.

Figure~\ref{fig:performance} shows the results for the Holevo variance for
$\widehat{\theta}\left(\vec{X}_\nu(\check{\theta}_1) ,\ldots,
  \vec{X}_\nu(\check{\theta}_m)\right)$ within $\theta \in [0, \pi/2)$ for
adaptive estimation strategies based on homodyne detection and squeezed vacuum
for different numbers of adaptive steps $m$ from $3$ to $15$. For all these
cases, we consider a strategy with a total number of copies of the probe state
$N = 3705$ with squeezing strength $r=1.01$, so that each adaptive step contains
$\nu=N /m$ probe states (These parameters were chosen for easy comparison with
previous works). The results shown in Figure~\ref{fig:performance} are obtained
from the average of five Monte Carlo simulations, each with $1\times10^4$ runs
of the strategy. The homodyne measurement without feedback and the two-step
adaptive homodyne strategy from Ref. \cite{Berni2015} are shown for comparison.
All the results have been normalized to the QCRB $=QFI^{-1}$. For these
simulations, sampling process employed the method of rejection sampling
\cite{robert1999monte}, while the optimization employed the method of
generalized simulated annealing over the interval $[0, \pi/2]$
\cite{xiang2013generalized}.

We observe that the proposed multi-step adaptive estimation strategy based on
homodyne detection consistently outperforms the non-adaptive and the two-step
homodyne strategies \cite{Berni2015}. Moreover, as can be observed in
Figure~\ref{fig:performance} (a) and the zoom in Figure~\ref{fig:performance}
(b), this adaptive homodyne strategy progressively approaches the QCRB (dashed
horizontal line), and is expected to saturate this bound for all values of
$\theta \in [0, \pi/2)$ in the limit of many adaptive steps. For instance, the
proposed adaptive estimation strategy with $m=15$ adaptive steps achieves a
precision of just $7\%$ above the QCRB for phases $\theta \in [0, \pi/2)$ on
average, compared to $42\%$ with two steps \cite{Berni2015}. We further note
that while the two-step strategy in \cite{Berni2015} shows a smaller variance
for $\theta \approx \theta_{\text{opt}}$ compared to the proposed strategy with
small $m = 3, 5$, its performance far from $\theta_{\mathrm{opt}}$ deviates
significantly from the QCRB, as can be seen in Figure~\ref{fig:performance} (b)
(see also Appendix~\ref{Appendix3}). Moreover, as discussed in the proof in
Appendix~\ref{Appendix_Proof}, by construction the proposed strategy ensures an
asymptotic quantum-limited performance. These results highlight the fundamental
advantage of this multi-step adaptive strategy for parameter estimation, and
underscore its potential for optical phase estimation and quantum metrology.

\begin{figure}[h!t]
  \centering \includegraphics[scale=0.32]{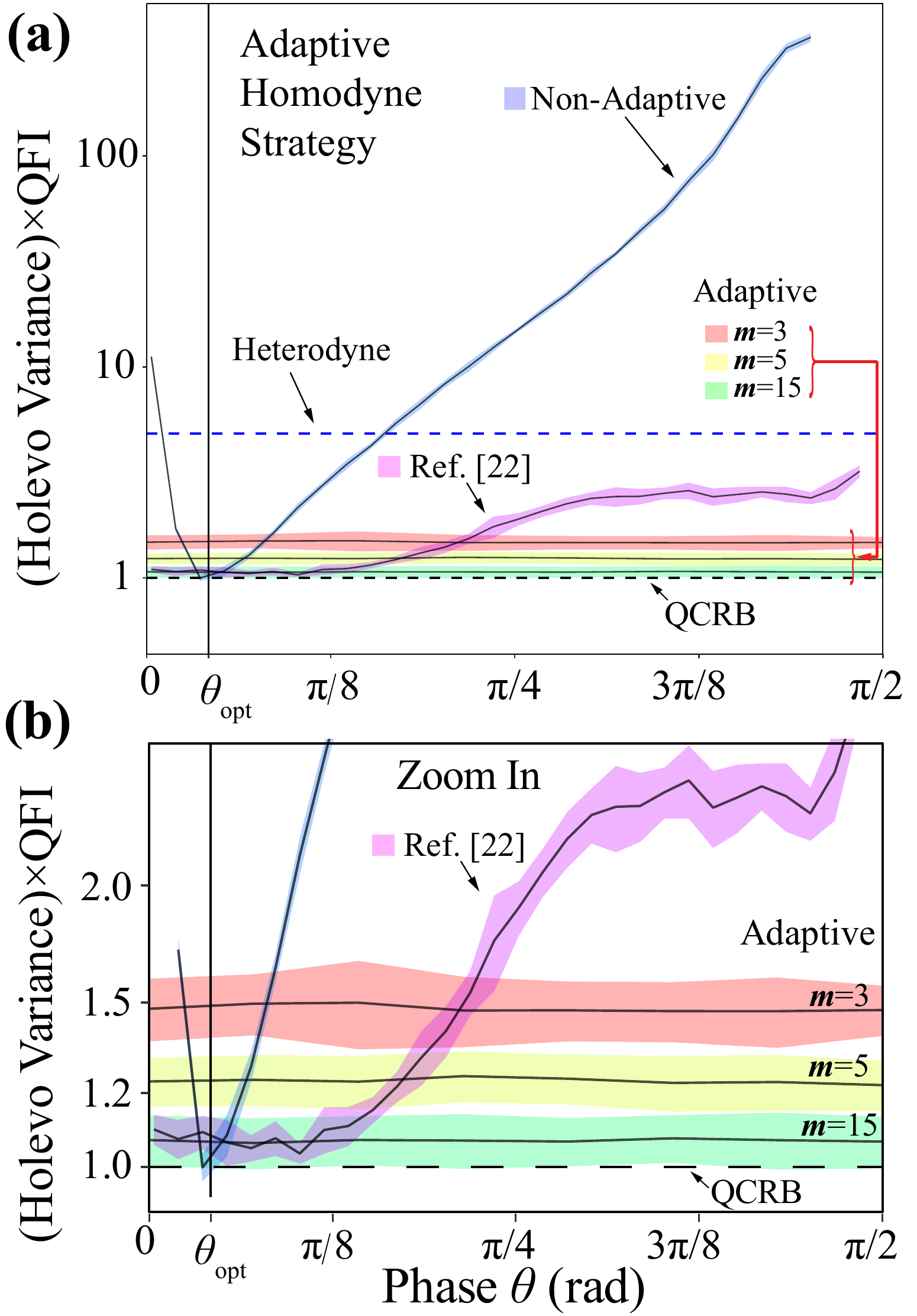}
  \caption{\label{fig:performance} \textbf{(a)} Holevo variance of the adaptive
    estimation strategy based on the AQSE formalism as a function of $\theta$,
    for different numbers of adaptive steps $m$, with $N=3705$ total independent
    copies of the probe state with squeezing strength of $r=1.01$. The y-axis
    shows the Holevo variance in logarithmic scale normalized to the inverse of
    the QFI in Eq.~(\ref{eq:QCRB_sq}), corresponding to the QCRB. The magenta
    line shows the two-step adaptive estimation scheme from Ref.
    \cite{Berni2015}, and the blue curve shows the estimation with homodyne
    detection without feedback. The shaded regions represent a one standard
    deviation. \textbf{(b)} Zoom in the region close to the QCRB.}
\end{figure}

\section{\label{sec:Losses} Loss and Noise Effects}
\subsection{Losses}

The performance of the adaptive estimation strategy for optical phase estimation
becomes sub-optimal in the presence of channel noise and loss, and when the
detectors have reduced quantum efficiency. Most of these effects can be modeled
as a mixing of the input probe with the vacuum state in a beam splitter. This
mixing process adds thermal photons to the probe incident into the homodyne
detector, depending on the transmission $T$ of the beam splitter
\cite{PhysRevA.79.033834, PhysRevA.81.033819}. Moreover, imperfect quantum
efficiency of the detectors can be modeled as a lossy channel, which further
contributes with additional thermal photons proportional to $T$
\cite{PhysRevA.79.033834,Lvovsky2015,RevModPhys.84.621}. Specifically, the lossy
channel modeled as a beam splitter with transmittance $0<T<1$ maps the squeezed
vacuum state $\rho_{r} = \lvert r, 0 \rangle \langle 0, r \rvert$ into a
squeezed thermal state
\begin{equation}
  \label{eq:sq_thermal}
  \rho_{\beta, r_l} = S(r_l) \left[ (1 -
    \mathrm{e}^{-\beta})\sum_{n=0}^{\infty} \mathrm{e}^{-\beta n} \lvert n \rangle
    \langle  n \rangle \rvert \right] S^{\dagger}(r_l),
\end{equation}
where $r_l < r$. This transformation reduces the QFI compared to that squeezed
vacuum states $\rho_{r_l}$ \cite{PhysRevA.79.033834, PhysRevA.81.033819,
  J_Twamley_1996} as:
\begin{equation}
  \label{eq:QFI_st}
  F_Q^{\text{Lossy}} = \left[  \frac{T^2}{1+2T(1-T) \sinh^2(r)  } \right] F_Q.
\end{equation}
where $F_Q = 2 \sinh^2(2r)$ is the QFI about $\theta$ for squeezed vacuum states
with strength $r$. Therefore, the effect of linear losses on the QCRB can be
effectively accounted for by an appropriate rescaling \cite{PhysRevA.79.033834}.

The squeezed thermal states resulting from losses further reduce the maximum
classical Fisher information for the homodyne measurement
\cite{PhysRevA.79.033834}:
\begin{equation}
  \label{eq:FI_st}
  F^{\text{Lossy}}_X(\theta_{\text{opt}}) = \left[  \frac{T^2}{1+4T(1-T) \sinh^2(r)  } \right] F_Q.
\end{equation}
Here, we have used the fact that the maximum of the Fisher information
$F^{\text{Lossy}}_X(\theta)$ is achieved at the optimal phase
$\theta_{\text{opt}}$ in Eq~\eqref{eq:theta_opt}, now for $r_l$.
\begin{figure}[h!t]
  \centering
  \includegraphics[scale=0.345]{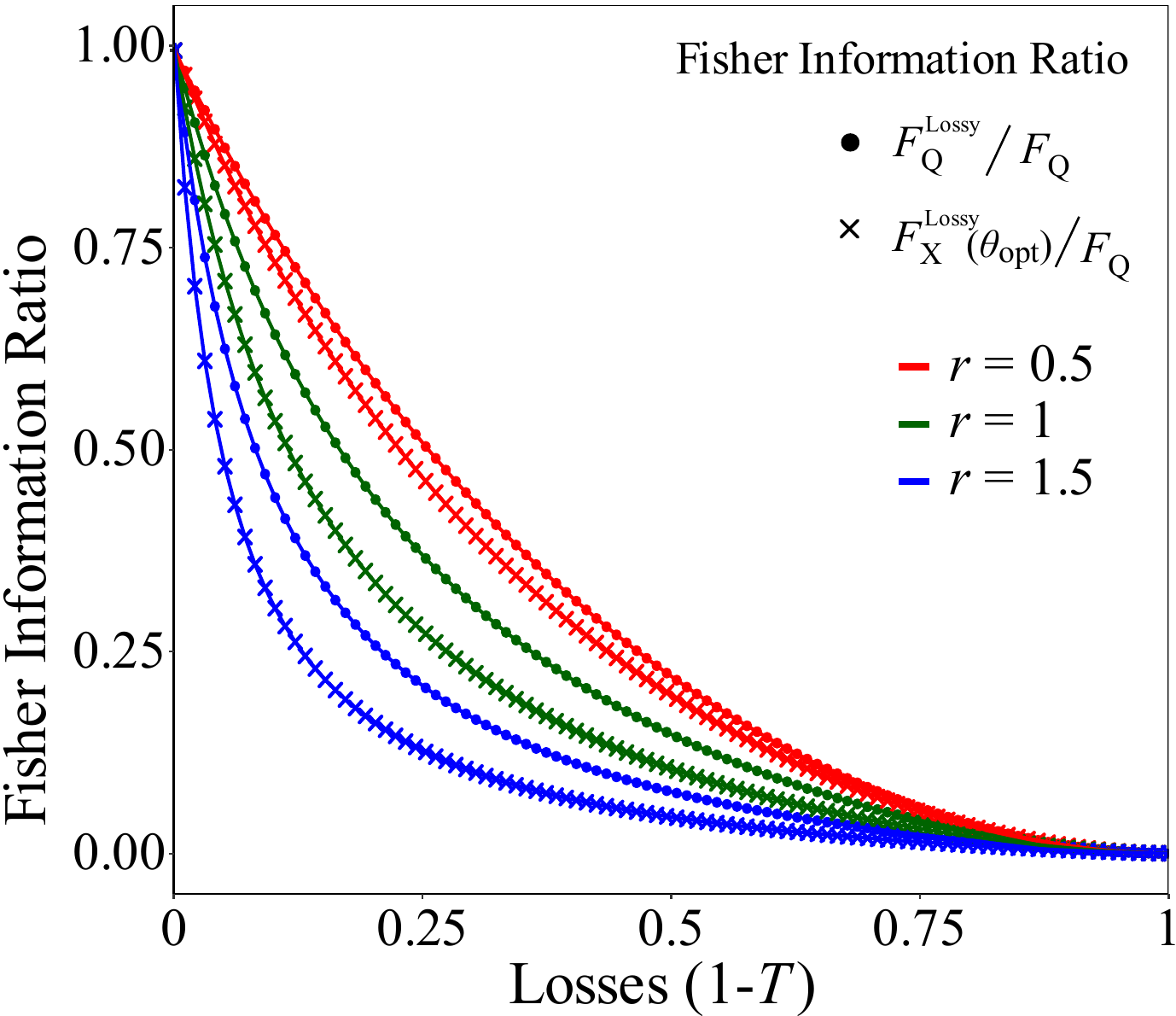}
  \caption{\label{fig:loss} Fisher information ratios $F^{\text{Lossy}}_Q/F_Q$
    (dots) and $F^{\text{Lossy}}_X(\theta_{\text{opt}})/F_Q$ (crosses) for squeezed
    vacuum states as function of $T$ for squeezing strength $r$ of $0.5$, $1.0$
    and $1.5$.}
\end{figure}

Figure~\ref{fig:loss} shows the ratios $F^{\text{Lossy}}_Q/F_Q$ (dots) and
$F^{\text{Lossy}}_X(\theta_{\text{opt}})/F_Q$ (crosses) as a function of the
losses ($1-T$) and different values of $r$. We observe that probe states with
larger $r$, are more sensitive to losses, showing a faster reduction in Fisher
information with channel loss. Figure~\ref{fig:loss} also shows that the Fisher
information $F_Q^{\text{Lossy}}$ and $F_X^{\text{Lossy}}(\theta_{\text{opt}})$
approach $F_Q$ only at $T\approx1$. We note that losses ($T<1$) reduce the
homodyne maximum Fisher information $F^{\text{Lossy}}_X(\theta_{\text{opt}}) <
F_Q^{\text{Lossy}}$, as seen in Eq. (\ref{eq:FI_st}), which prevents the
saturation of the QCRB solely with homodyne. Moreover, we note that in general
the saturation of the QCRB for squeezed thermal states requires the
implementation of non-Gaussian measurements \cite{Oh2019}. However, devising
non-Gaussian measurements that approach the QFI for this problem are highly
complex, and their implementation is still an open problem \cite{Oh2019}.

In a more general setting, the problem of quantum channel estimation involves
finding the optimal probe states and measurements to maximize information in a
lossy and noisy channel \cite{Liu2020, Sisi2024}. Recent advances in quantum
channel estimation have found optimal probes for: Gaussian unitary channels
(which consist of rotated squeezed states) \cite{PhysRevA.94.062313}, bosonic
dephasing channels \cite{MW2024}, and complete positive trace preserving (CPTP)
maps for qubits \cite{Sisi2024}. Moreover, in principle, it is possible to find
the optimal quantum probe states for a general channel estimation problem using
techniques of convex optimization including semidefinite and linear programming,
and conic programming \cite{PhysRevA.107.012428,Sisi2024,Hayashi2024}. Here, we
focus our discussion on squeezed probe states, and adaptive homodyne/heterodyne
measurements, which are readily available in laboratory settings, and consider
common sources of noise and imperfections, including linear losses.

\subsection{Imperfect state preparation}

The processes of state preparation in realistic implementations is often
affected by small random errors yielding state preparation errors. Due to the
Central Limit Theorem, the cumulative effect of small, independent errors in
many probe states will tend to a normal distribution. Therefore, it is
reasonable to assume that, under some state preparation errors, the squeezing
strength $r$ of the probe state $\rho_r = \lvert r, 0 \rangle \langle 0, r
\rvert$ is the output of a random variable $R$ with a normal distribution
$\mathcal{N}\left( r_0, \sigma_r^2 \right)$, where $\sigma_r^2 \geq 0$ is a
small number relative to $r_0$.

To take into account state preparation errors of this kind in the adaptive
strategy, we consider that the samples at every adaptive step are drawn from the
conditional random variable $X \mid R$, where $X$ is a sample from a homodyne
measurement. We then analyze the strategy as described in
Section~\ref{sec:Op_dyne} but considering state preparation errors. In this
case, and without loss of generality by taking $X$ as a sample of
$M_{\text{Hom}}$, the Fisher information about $\theta$ contained in the
conditional random variable $X \mid R$ is calculated with respect to the
conditional density of $X$ given $R$, that is,
\begin{equation}
  \label{eq:QFI_cond}
  F_{X \mid R}(\theta) = \mathrm{E}\left[ F_{X \mid R=r} (\theta) \right],
\end{equation}
where $F_{X \mid R=r} (\theta)$ corresponds to Eq.~\eqref{eq:Fisher_Homodyne} at
a specific value $R=r$. Therefore, in the asymptotic limit, the adaptive
homodyne strategy samples around the phase that maximizes
Eq.~\eqref{eq:QFI_cond}. In experimental settings, typically the range for
$\sigma_r$ lies between $0.01$ and $0.02$ \cite{Berni2015,Xu:19}. For state
preparation errors with small variances $\sigma_r^2$, the optimal phase
$\theta_{\mathrm{opt}}^{noise}$ deviates slightly from the noiseless case
$\theta_{\mathrm{opt}}$, and modifies the Fisher information.

Figure~\ref{fig:state_imp} shows $F_{X \mid R}(\theta)$ for standard deviations
of $r$ in state preparation $\sigma_r=0.01$ and $\sigma_r=0.02$ as a function of
$\theta \in [0, \pi/2)$. We observe that errors in state preparation with
$\sigma_r$ of $0.01$ and $0.02$ have a negligible effect in the performance of
the adaptive strategy, which shows a Fisher information around the QFI at the
optimal phase $\theta_{opt}$. We also note that since the QFI for squeezed
vacuum is a nonlinear function of the squeezing strength $r$, the contribution
of positive deviations of $r$ from the mean $r_0$ to the expected value of $F_{X
  \mid R}(\theta)$ in Eq.~(\ref{eq:QFI_cond}) increases with $r$. This nonlinear
effect causes the maximum of $F_{X \mid R}(\theta)$ to slightly surpass $F_{X
  \mid R = r}(\theta_{\mathrm{opt}})$ corresponding to the QFI (dashed red
line), as can be observed in the inset ($i$) of Figure~\ref{fig:state_imp}.

\begin{figure}[h!t]
  \centering
  \includegraphics[scale=0.35]{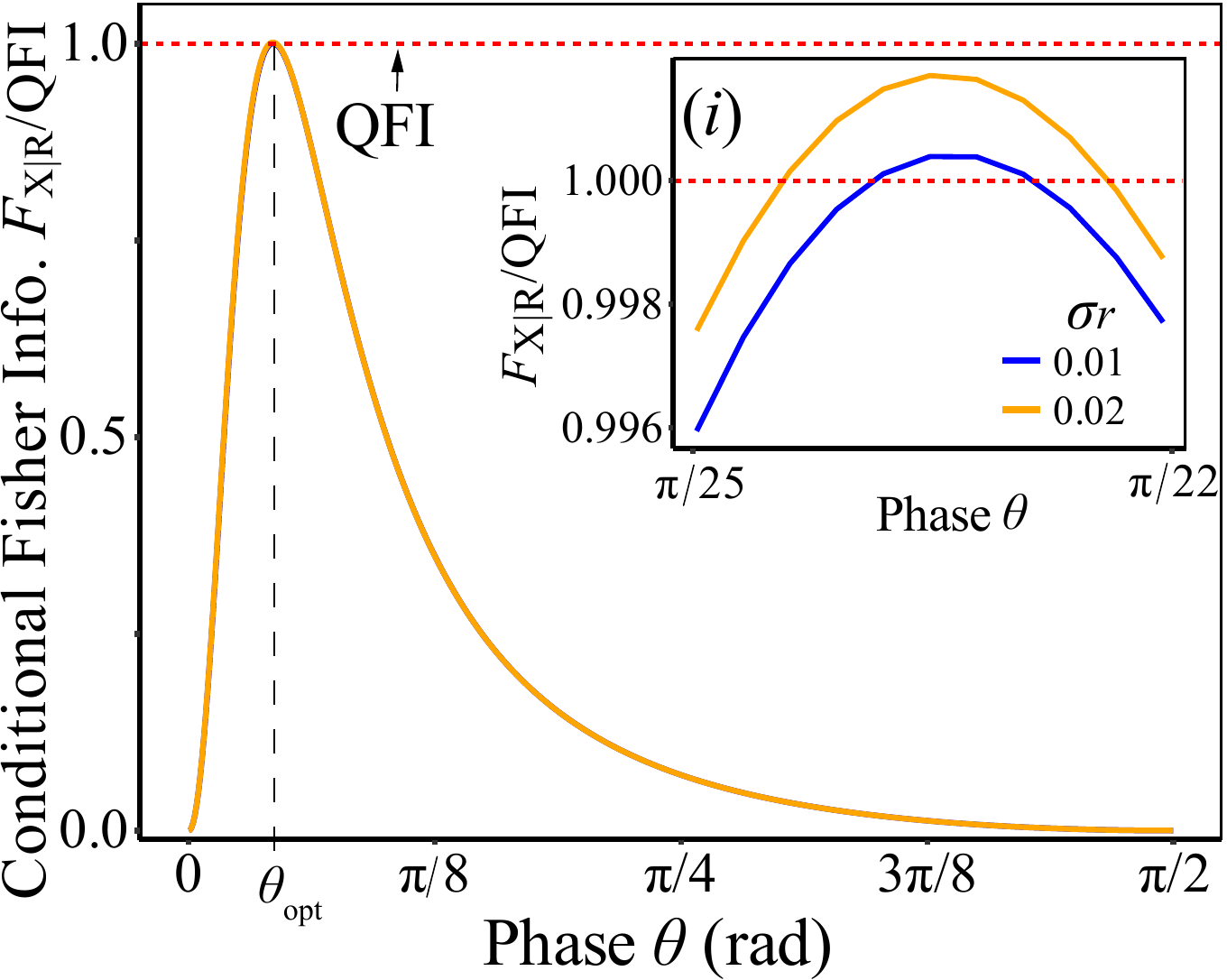}
  \caption{\label{fig:state_imp} Expected Fisher information as function of
    $\theta$ with state preparation errors with normally varying squeezing
    strength for $r_0 = 1$ with $\sigma_r$ of $0.01$ and $0.02$. The dashed red
    line corresponds to the QFI for squeezed vacuum states at $r=1$. Inset ($i$)
    shows a zoom in the maximum of the curve showing the overshot effect due to
    the nonlinear dependence of $F_{X \mid R}(\theta)$ with $r$.}
\end{figure}

\subsection{Phase errors in the homodyne local oscillator (LO)}

As a case of study for measurement errors, we consider that the LO in the
homodyne measurement is subject to small phase errors. Thus, the parameter
$\check{\theta}$, describing the phase estimate and the setting for the POVM
$M_{\check{\theta}}$ in the adaptive measurement protocol, can be considered as
the output of a random variable $\check{\Theta}_{\text{lo}}$ with a normal
distribution $\mathcal{N}\left( \check{\theta}_{\text{lo}},
  \sigma^2(\check{\theta}_{\text{lo}}) \right)$, centered at the ideal
measurement setting $\check{\theta}_{\text{lo}}$ and with a variance
$\sigma^2(\check{\theta}_{\text{lo}})$. Assuming that
$\check{\theta}_{\text{lo}} \approx \theta$, we can evaluate the loss of
information caused by these kinds of measurement errors by calculating the
conditional Fisher information with respect to the random variable $X \mid
\check{\Theta}_{\text{lo}}$,
\begin{equation}
\label{eq:QFI_cond2}
F_{X \mid \check{\Theta}_{\text{lo}}  } = \mathrm{E}\left[ F_{X \mid  \check{\Theta}_{\text{lo}} = \check{\theta}} \right],
\end{equation}
where $X$ is a sample from the POVM $M_{\check{\theta}}$ applied to the state
$\rho_r = \lvert r, 0 \rangle \langle 0, r \rvert$.

Figure~\ref{fig:state_imp2} shows Eq~\eqref{eq:QFI_cond2} as a function of the
standard deviations $\sigma(\check{\theta}_{\mathrm{lo}})$ in measurement
implementation, ranging from $0.01$ to $0.1$ radians. We observe that errors in
the phase reference for the homodyne detection has a moderate detrimental effect
in the Fisher information, which decreases approximately to half of the QFI for
large LO phase noise of $\sigma(\check{\theta}_{\mathrm{lo}})\approx 0.15$ rad.
However, we note that this technical problem can be overcome by standard phase
stabilization techniques \cite{DiMario2022,PhysRevResearch.3.013200}.

\begin{figure}[h!t]
  \centering
  \includegraphics[scale=0.33]{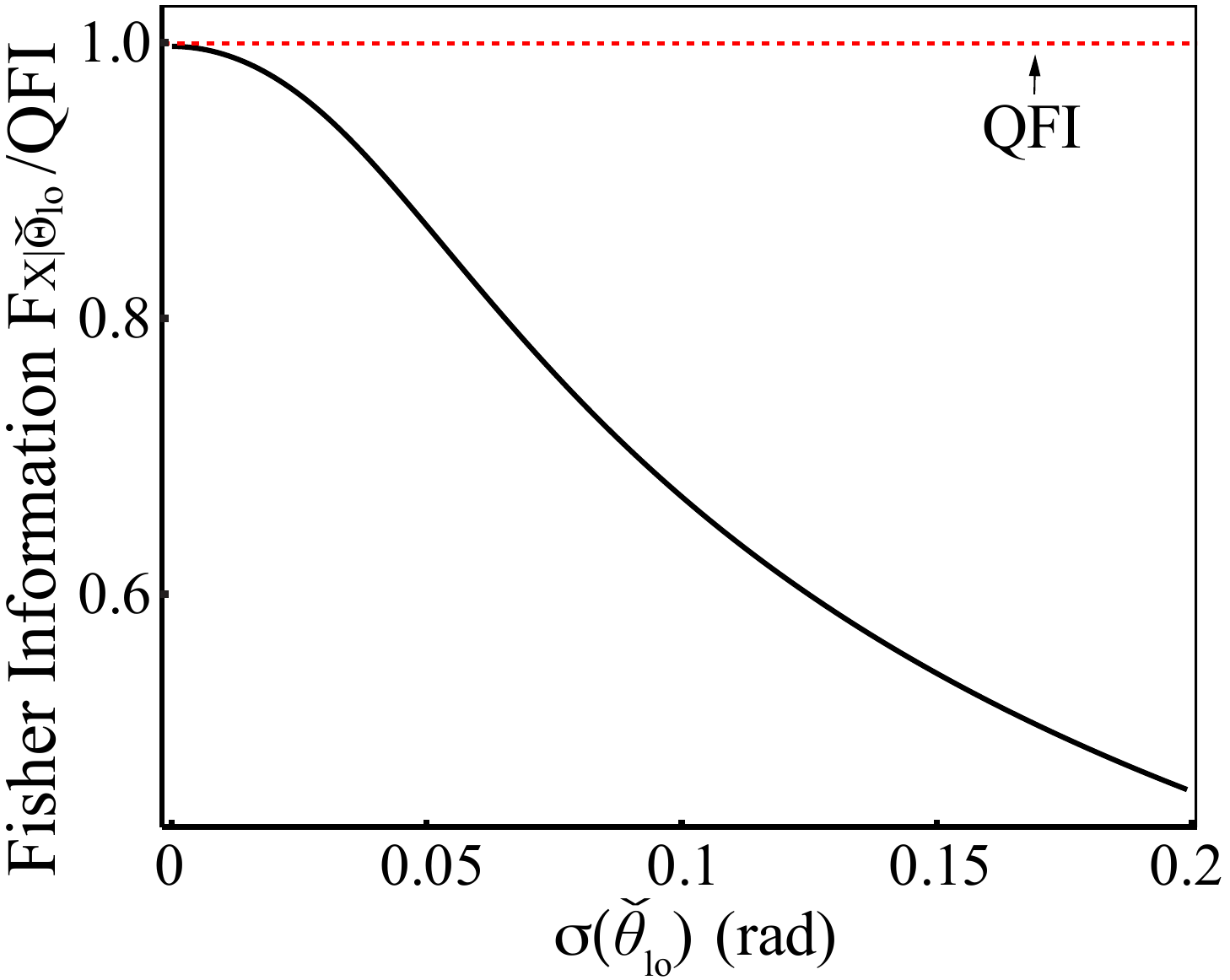}
  \caption{\label{fig:state_imp2} Expected Fisher information of the homodyne
    measurement as a function of $\sigma(\check{\theta}_{\mathrm{lo}})$
    quantifying the phase errors in the LO for input states with a squeezing
    strength of $r=1$ and measurement errors normally distributed around
    $\check{\theta}_{\mathrm{lo}} = \theta$. The dashed red line
    corresponds to the QFI for squeezed vacuum states at $r=1$. }
\end{figure}

\section{\label{sec:General-dyne} Combined homodyne-heterodyne (CHH) measurement
  strategy}

In general, the problem of phase estimation involves estimation over the
complete range of possible phases from $0$ to $2\pi$. However, when using
squeezed vacuum states for phase estimation beyond the SNL, there is a physical
limitation on the range of phases that can be estimated. Squeezed vacuum states
are invariant under phase shifts of $\pi$, restricting the estimable phases to
the interval $[0, \pi)$, which is half of the complete range in the general
problem of phase estimation. This symmetry can be seen in the Husimi Q
representation of the state $\rho_r$, $Q(\alpha) =
\frac{1}{\pi}\bra{\alpha}\rho_r\ket{\alpha}$ (see inset ($i$)
Fig.~\ref{fig:scheme}), which shows that the squeezed vacuum probe can only
encode the phase modulo $\pi$. Moreover, the measurement employed for decoding
the phase can impose severe constraints on the range of phases that can be
estimated. For instance, homodyne measurements further reduce the range of
phases within which phase estimation is possible to $[0,\pi/2)$. This is because
the probability distributions of outcomes from homodyne measurements in
Eq.~\eqref{eq:pdf_homodyne}, associated with POVMs in Eq.~\eqref{eq:set_of_lues},
are $\pi/2$ periodic. Consequently, any strategy based on adaptive homodyne is
restricted to estimating phases within the range $[0, \pi/2)$. To go beyond this
limited range and enable phase estimation with squeezed vacuum over the entire
range of phases $[0, \pi)$, it is necessary to include measurements beyond
homodyne.

\begin{figure}[h!t]
  \centering
  \includegraphics[scale = 0.31]{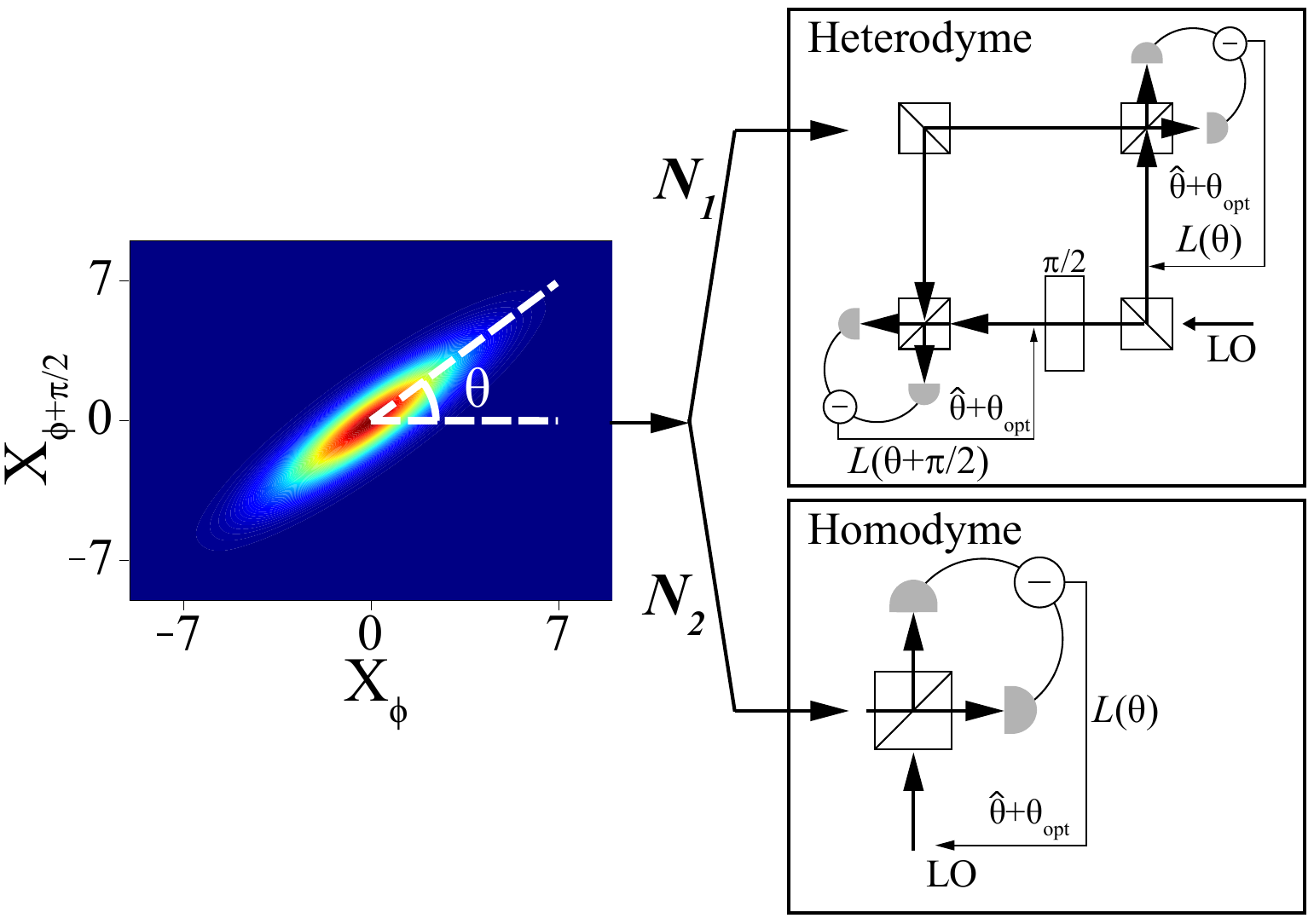}
  \caption{\label{fig:scheme_generalized} Combined homodyne-heterodyne (CHH)
    measurement strategy for phase estimation based on squeezed vacuum probes
    for phases $\theta\in[0, \pi)$. This strategy takes advantage of the
    capability of heterodyne measurements to unambiguously estimate phases
    within the whole parametric space $[0, \pi)$ for squeezed vacuum, overcoming
    the non-identifiability problem in the likelihoods from homodyne
    measurements. By employing a small sample of heterodyne measurements, the
    unknown phase $\theta$ is localized within a neighborhood within $[0, \pi)$.
    Then, the strategy employs adaptive homodyne for phase estimation within
    this neighborhood.}
\end{figure}

We propose a combined homodyne-heterodyne (CHH) measurement strategy that uses
heterodyine to identify the neighborhood of the unknown phase within $[0,\pi)$,
and subsequently, adaptive homodyne to implement an asymptotically optimal
measurement strategy. This CHH measurement uses time sharing between heterodyne
and homodyne, which is a special case of the generalized dyne measurement
\cite{Oh2019}, and enables optimal phase estimation within $[0,\pi)$ in the
asymptotic limit.

Figure \ref{fig:scheme_generalized} shows the schematic of the proposed CHH
strategy. The strategy implements a heterodyne measurement on a small sample of
$N_1$ probe states, denoted as $\vec{X}^{\text{Het}}_{N_1}$, to determine the
neighborhood in the parametric space $[0, \pi)$ to which the unknown phase
belongs. This makes the parameter identifiable within $[0,\pi)$ solving the
non-identifiability problem of homodyne \cite{Hayashi2005}, and produces a
likelihood peaked around the true value. After the heterodyne sampling
$\vec{X}^{\text{Het}}_{N_1}$, the CHH strategy implements the adaptive homodyne
strategy described in Section \ref{sec:Op_dyne}, with $N_2$ copies of the probe
state and $m$ adaptive measurements. By construction, when $N_2,m \to \infty$
(in the asymptotic limit), and $N_1$ is large enough such that the MLE from the
heterodyne sampling $\widehat{\theta}\left(\vec{X}^{\text{Het}}_{N_1}\right) \in
[0, \pi/2)$ with high probability, this strategy saturates the QCRB
Eq.~\eqref{eq:QCRB_sq} \cite{Hayashi2005,Paninski2005} (see
Appendix~\ref{AppendixExtension} for discussion of the proof). Thus, in the
asymptotic limit this strategy is expected to be able to extract all the
information about the phase encoded in the quantum probe and approach the QCRB
for any phase within $[0, \pi)$.

Figure~\ref{fig:performance_2} (a) shows the performance of the CHH phase
estimation strategy enabling near quantum-optimal phase estimation for phases
within $[0,\pi)$ with a finite number of samples. These results are obtained
from the average of five Monte Carlo simulations considering $N=3705$ copies of
the squeezed vacuum probes with $m=15$ adaptive steps, each with a sample of
size $\nu=247$. In this combined strategy, the first step consists of a
heterodyne measurement with a sample of size $N_1 = \nu_{1} = 247$. This sample
is large enough to produce estimates with high probability in $[0, \pi/2)$, and
it is significantly smaller than the total of subsequent (homodyne) samples $N_2
= \sum_{i=2}^{14}\nu_i = 3458$. We note that within the statistical noise of our
simulations, the CHH strategy enables phase estimation approaching the QCRB
within the full interval $\theta \in [0, \pi)$, as seen in the zoom in
Figure~\ref{fig:performance_2} (b). For a more rigorous analysis of the
convergence of the CHH strategy Appendix \ref{AppendixExtension} discusses the
proof of the asymptotic convergence of the CHH strategy to the QCRB over the
full range $[0,\pi)$.

\begin{figure}[h!t]
  \centering
  \includegraphics[scale=0.37]{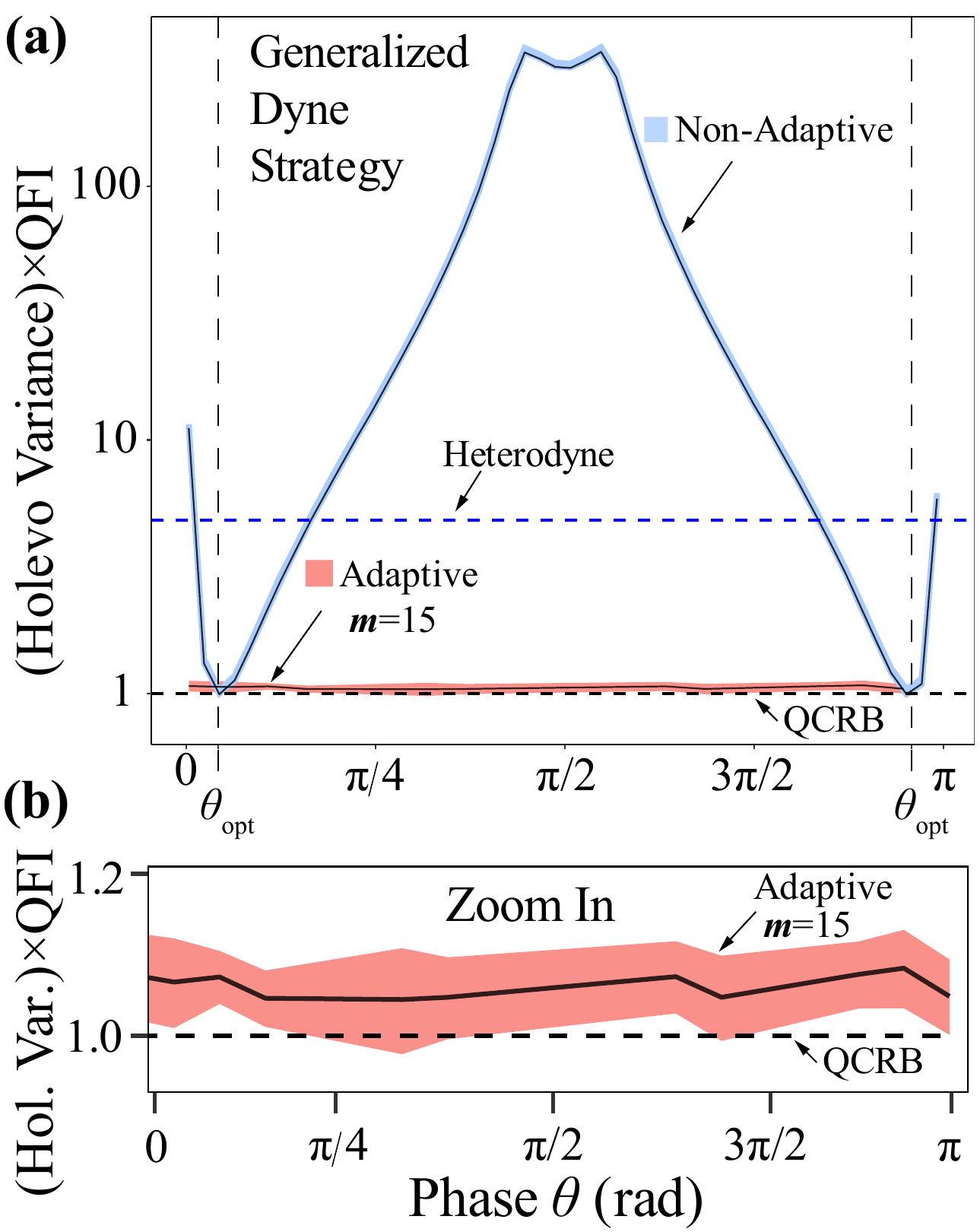}
  \caption{\label{fig:performance_2} \textbf{(a)} Performance of the CHH phase
    estimation strategy with $m=15$ adaptive steps (orange), using $N=3705$
    probes with $r=1.01$. The performance of homodyne detection without feedback
    (blue) is shown for reference. In the CHH strategy, the initial adaptive
    step involves heterodyne sampling, while the subsequent adaptive steps
    utilize locally optimal homodyne POVMs in Eq.~(\ref{eq:set_of_lues}).
    Throughout the simulation, the sample size remains constant at $\nu = 247$
    (number of probe states per adaptive step). The dashed blue line shows the
    heterodyne limit, while the dashed black line corresponds to the QCRB. The
    shaded regions indicate a one standard deviation. \textbf{(b)} Zoom
    in the region close to the QCRB.}
\end{figure}

\begin{figure}[h!t]
	\centering
	\includegraphics[scale=0.32]{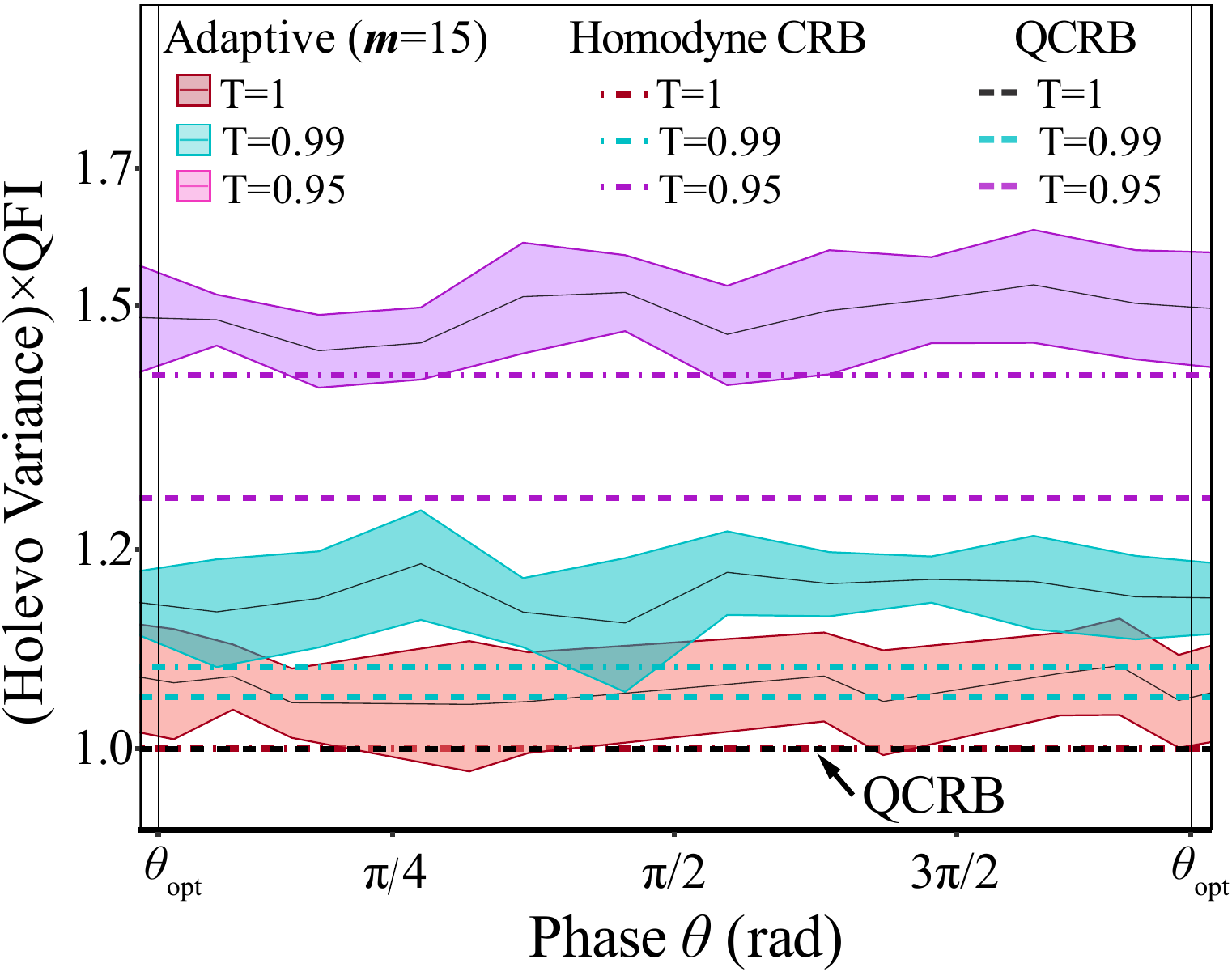}
	\caption{\label{fig:performance_2NOISE}Performance of the CHH strategy under
    losses, with $m=15$ adaptive steps and $N=3705$ probes with $r=1.01$. Losses
    are characterized by the channel transmission $T$ for $T=1$ in brown
    (lossless), $T=0.99$ in light blue, and $T=0.95$ in purple. The
    dotted-dashed lines represent the minimum CRB for the homodyne measurement
    normalized to the QCRB, i.e.
    $F^{\text{Lossy}}_Q/F^{\text{Lossy}}_X(\theta_{\text{opt}})$. The dashed
    lines represent the QCRB for different channel transmissions $T$. For $T=1$,
    the homodyne CRB equals the QCRB due to the local optimality of the homodyne
    measurement.}
\end{figure}

As a final step, we investigate the CHH measurement strategy under channel
losses. Fig.~\ref{fig:performance_2NOISE} shows the normalized Holevo variance
of the estimator obtained through the CHH estimation strategy with $m=15$
adaptive steps each with $\nu=247$ probe states with a squeezing strength
$r=1.01$. The loss is characterized by the channel transmission $T$ for $T=1$ in
brown (lossless), $T=0.99$ in light blue, and $T=0.95$ in purple. The shaded
regions represent one standard deviation. The dotted-dashed lines represent the
minimum Cramér-Rao bound (CRB) for the homodyne measurement for different
channel transmissions $T$, obtained as the inverse of the classical Fisher
information in Eq.~\eqref{eq:FI_st}, and normalized to the QCRB. The dashed
lines represent the QCRB for different channel transmissions $T$, which is
obtained as the inverse of Eq.~\eqref{eq:QFI_st}. We note that for a lossless
channel $T=1$, the CRB equals the QCRB due to the local optimality of homodyne
measurements. We observe that the performance of the CHH strategy is close to
the homodyne CRB (dotted-dashed) for the channel transmissions considered, and
it is expected to approach this bound for all cases. Moreover, the performance
of the CHH strategy is maintained for all the phases within the full range
$[0,\pi)$ for squeezed vacuum.

\section{\label{sec:Diss}  Discussion \& Conclusions}

Our analysis and numerical simulations show that the proposed adaptive
estimation strategy efficiently extracts the maximum attainable information
about the unknown phase encoded in squeezed vacuum states in the asymptotic
limit. However, we note that while the proposed strategy allows for phase
estimation at the quantum limit for the full range of phases $[0, \pi)$ for
squeezed vacuum states, this range is limited due to the $\pi$ phase-shift
symmetry inherent to these states. We also note that other quantum states used
for phase estimation at the quantum limit such as NOON states face the same
limitation due to their inherent phase-shift symmetries
\cite{escher2011quantum}. On the other hand, there may be other optical probes
capable of solving the non-identifiability problem in the phase encoding due to
state symmetries, albeit with lower QFI. Coherent states, for example, have a
significantly lower QFI compared to squeezed vacuum states, but allow for
unambiguously identifying the quadrants of the phase within $[0,2\pi)$
\cite{PhysRevLett.128.230501, PhysRevLett.125.120505, Rodriguez-Garcia2022,
  PhysRevA.70.043812}.

As an alternative quantum probe for phase estimation, displaced squeezed states
$D(\alpha) \lvert 0, r \rangle$ can offer the ability to unambiguously encode
phases within $[0, 2\pi)$ with a higher QFI compared to coherent states. We
note, however, that there will be a trade-off between the achievable QFI
compared to that of squeezed vacuum states and the ability to identify the
quadrants of the phase. Finding the best trade off requires optimization of both
the squeezing strength $r$ and the displacement parameter $\alpha$, given a
fixed resource budget in terms of the number of photons. Moreover, this
trade-off will critically depend on the available resources and experimental
constraints. Further research will focus on exploring and identifying the
optimal quantum probe states capable of overcoming the non-identifiability
problem while maintaining a high QFI for phase estimation.

In summary, we propose a Gaussian estimation strategy for optical phase
estimation with squeezed vacuum states that approaches the quantum limit in
precision. This strategy leverages homodyne measurements and rotations to
implement a complete set of locally optimal POVMs. This set of POVMs are used to
construct an adaptive estimation method based on the Adaptive Quantum State
Estimation (AQSE) formalism, which ensures consistency and efficiency of the
estimator in the asymptotic limit, with variance equal to the inverse of the QFI
for phases $\theta\in[0, \pi/2)$. To extend the parameter range for phase
estimation to $[0, \pi)$, which is the maximum range of phases that can be
encoded in squeezed vacuum states, we generalize the estimation strategy to
incorporate a small number of heterodyne measurements. This heterodyne sampling
allows for identifying the neighborhood of the phase within $[0, \pi)$, solving
the non-identifiability problem in the likelihoods from homodyne measurements,
while maintaining a quantum-optimal performance in the limit of many adaptive
steps. This result represents a significant advancement in high-precision
quantum metrology and optical phase estimation based on quantum correlated
states.

\section{\label{sec:Ap} Appendix}

\subsection{\label{sec:MLE} Efficiency and consistency of MLE in quantum
  systems}

This appendix describes the conditions under which the MLE used in AQSE is
consistent and efficient over the complete parameter space $\Theta$. Let us
consider the problem of estimating an unknown parameter $\theta \in \Theta$
associated with a set of quantum states $\left\{ \rho(\theta) : \theta \in
  \Theta \right\}$ from measurements of the system. When independent
measurements are performed over the system, a set of independent random
variables $\vec{X}_N = X_1, \ldots, X_N \in \mathcal{X}^N, \, N \geq 1$ carry
the information about $\theta$. In this case, the total Fisher information about
$\theta$ is the sum of the individual Fisher information values for each
measurement. This property can be exploited to reach the QCRB in the asymptotic
limit ($N \to \infty$). When the outcomes of a POVM $M$ have a Fisher
information that coincides with the QFI, and their probability distribution
satisfies a set of mild regularity conditions described below, it can be shown
that, in the asymptotic limit, the MLE applied to the outcomes of $M$ can achieve
the QCRB \cite{Fujiwara2006}.

To saturate the QCRB, the MLE requires to be asymptotically consistent, which
means that as the sample size increases, the MLE converges to the true value of
the parameter $\theta$ in probability (weak sense) or almost surly (strong
sense) \cite{Keener2010}. For a MLE to be asymptotically consistent, the
following conditions over the parametric set $\Theta$ and the set of density
functions $\left\{ f(X_i \mid \theta; M_i) \right\}_{\theta \in \Theta}$ for
each POVM $M_i$ must be satisfied \cite{Keener2010,Paninski2005, Fujiwara2006}:
\begin{itemize}
\item \textbf{Compactness:} The parameter space $\Theta$ and the space of POVMs
  must be compact, which means that it is closed and bounded. This property
  ensures that the MLE exists for any sample size.
\item \textbf{Identifiability:} The true value of the parameter must be uniquely
  determined by the probability distribution. In other words, different values
  of the parameter must produce different probability distributions.
\item \textbf{Measurability:} The probability density function $f(X_i \mid
  \theta; M_i)$ must be measurable for all $X_i = x_i$ and for each POVM $M_i$.
  Thus the MLE is well-defined as a random variable.
\item \textbf{Continuity:} The probability density function $f(X_i \mid \theta;
  M_i)$ must be continuous in the parameter space $\Theta$ for all $X_i = x_i$
  and for each POVM $M_i$. This guarantees that small changes in the value of
  the parameter result in small changes in the probability.
\item \textbf{Dominance:} The log likelihood $\log \left[ f(X_i \mid \theta;
    M_i) \right]$ is uniformly Lipschitz in $\theta$ with respect to some
  dominating measure on $\mathcal{X}$. This provide the convergence of the MLE.
\end{itemize}

Under this set of regularity conditions, the MLE exhibits asymptotic
consistency. As a result, assuming sufficiently smooth likelihoods, the
distribution of the MLE in the limit $N \to \infty$ follows a normal
distribution:
\begin{equation}
  \label{eq:asymp_normal}
  \widehat{\theta}(\vec{X}_N) \sim \mathcal{N}\left(\theta, \frac{1}{F_{\vec{X}_N}(\theta)}  \right).
\end{equation}
Here $\widehat{\theta}(\vec{X}_N)$ denotes the MLE based on the sample
$\vec{X}_N$, and $F_{\vec{X}_N}(\theta)$ represents the Fisher information
associated with the sample. Consequently, the variance of
$\widehat{\theta}(\vec{X}_N)$ is given by $\frac{1}{F_{\vec{X}_N}(\theta)}$, and
$\widehat{\theta}(\vec{X}_N)$ achieves the classical Cramér-Rao bound for all
$\theta \in \Theta$. Notably, when the Fisher information of the random
variables $\vec{X}_N$ equals the QFI, the MLE attains the QCRB, which
corresponds to the ultimate limit in precision for parameter estimation. Based
on this observation, the AQSE framework can be exploited to construct an
asymptotically optimal strategy using homodyne detection. By satisfying the
regularity conditions outlined above, the MLE is guaranteed to converge to the
true value, and AQSE can be used to adapt the homodyne measurements to sample
around the optimal point $\theta_{\text{opt}}$, attaining the QCRB.

\subsection{\label{Appendix_Proof} Proof of Convergence of the phase estimator
  variance to the QCRB}

\subsubsection{\label{Appendix_Proof_s1} Asymptotic consistency of the MLE for
  the adaptive homodyne strategy}

This appendix details the proof of the strong asymptotic consistency of the MLE
obtained from the proposed adaptive homodyne strategy based on AQSE. As a first
step, the adaptive strategy defines an initial estimate $\check{\theta}_1 \sim
U\left( \tilde{\Theta} \right)$ for the parameter $\theta$, where
$\tilde{\Theta}$ is a compact subset of $(0, \pi/2)$. Then, the homodyne
measurement with POVM $M_{\check{\theta}_1}$ in Eq.~\eqref{eq:lue_homodyne}
yields a sample of measurement outcomes of size $\nu$,
$\vec{x}_{\nu}(\check{\theta}_1) = x_1(\check{\theta}_1), \ldots
x_\nu(\check{\theta}_1)$. The MLE is then applied to this sample of measurement
outcomes $\vec{x}_{\nu}(\check{\theta}_1)$ resulting in an updated estimate
$\check{\theta}_1 :=
\widehat{\theta}_{\mathrm{MLE}}\left(\vec{x}_{\nu}(\check{\theta}_1)\right)$.
This new estimate serves as the subsequent guess of the parameter $\theta$ for
the next adaptive step. Then, for a given subsequent adaptive step $m$, $m \geq
2$, the adaptive strategy performs the POVM $M_{\check{\theta}_{m}}$ yielding
the sample of outcomes $\vec{x}_{\nu}(\check{\theta}_{m}) =
x_1(\check{\theta}_m), \ldots x_\nu(\check{\theta}_m)$, from which the MLE
produces an estimate
$\widehat{\theta}_{\mathrm{MLE}}\left(\vec{x}_{\nu}(\check{\theta}_1),\ldots,
\vec{x}_{\nu}(\check{\theta}_{m}) \right) = \check{\theta}_{m+1}$. This procedure is
repeated iteratively in subsequent adaptive measurements. By satisfying the
regularity conditions described in Appendix:~\ref{sec:MLE} for the statistical
model for each homodyne measurement, we prove the \textit{almost sure}
convergence of the sequence of MLEs
$\widehat{\theta}_{\mathrm{MLE}}\left(\vec{X}_{\nu}(\check{\theta}_1), \ldots,
\vec{X}_{\nu}(\check{\theta}_m)\right)$ to the true parameter $\theta$ as the number
of adaptive steps $m \to \infty$.

First, we note that the set of homodyne measurements $\left\{
  M_{\check\theta}(dx) \right\}$ forms a set of locally optimal POVMs
parameterized by $\tilde{\Theta}$. Without loss of generality, we assume that
the true phase to be estimated is $\theta_0 \in \tilde{\Theta}$. For this proof,
we assume that the regularity conditions described in Appendix~\ref{sec:MLE} are
satisfied. We first present a series of auxiliary
\textbf{\Cref{obs1,obs2,obs3,obs4}}. Then we present the main result of this
proof in \textbf{Theorem \ref{MainTheorem}}, which shows the asymptotic
consistency and convergence of the MLE. The method employed in this proof is
analogous to the technique used to bound the probability of rejecting the null
hypothesis in binary hypothesis testing based on the likelihood ratio test
\cite{Keener2010}. This technique has been applied to prove the asymptotic
consistency to the MLE in the context of Optimal Design of Experiments
\cite{mccormick_strong_1988} and in the context of quantum parameter estimation
\cite{Fujiwara2006}.

Let $\phi(\theta, \check{\theta}) = \theta + \theta_{\text{opt}} -
\check{\theta}$ be the argument of the homodyne POVM in Eq.
(\ref{eq:lue_homodyne}) for any $\theta$, with $\check{\theta} \in
\tilde{\Theta}$ and $\theta_{\text{opt}}$ defined in Eq.~\eqref{eq:theta_opt}.
For any $\epsilon > 0$, let the open neighborhood centered at $\theta$ and
radius $\epsilon$ be:

\begin{equation*}
  N_{\epsilon}(\theta) = \left\{ \theta': \lvert \theta'-\theta \rvert < \epsilon \right\} = (\theta - \epsilon, \theta + \epsilon).
\end{equation*}
Let the log of the ratio of likelihoods be:
\begin{equation*}
  R(\theta, \theta_0, \check{\theta}) = \log \left( \frac{f(\vec{X}_{\nu}(\check{\theta}) \mid \phi(\theta, \check{\theta})    )}{f(\vec{X}_{\nu}(\check{\theta}) \mid \phi(\theta_0, \check{\theta})   )} \right),
\end{equation*}
and
\begin{equation*}
  R(\epsilon, \theta, \theta_0, \check{\theta}) =  \sup_{\theta' \in N_{\epsilon}(\theta) }   R(\theta', \theta_0, \check{\theta}),
\end{equation*}
where $f(x \mid \theta)$ is defined in Eq.~\eqref{eq:pdf_homodyne} for any
$\theta \in \tilde{\Theta}$.

\begin{lemma}
  \label{obs1}
  For any $\theta, \theta_0 \in \tilde{\Theta}$ with $\theta \neq \theta_0$, it
  follows that:
 \begin{equation}
   \label{eq:lemma1}
   \left[  \frac{ 2 \sigma\left( \phi(\theta, \check{\theta})  \right)  \sigma\left(   \phi(\theta_0, \check{\theta})  \right)  }{ \sigma^2\left( \phi(\theta, \check{\theta})  \right) +  \sigma^2\left( \phi(\theta_0, \check{\theta})  \right)   } \right]^{1/2} < 1,
 \end{equation}
 where $\sigma(\theta)$ is defined in Eq.~\eqref{eq:var_homodyne} for the
 probability density function in Eq.~\eqref{eq:pdf_homodyne} of the outcomes of
 homodyne measurements.
 \\
\end{lemma}

\begin{proof}
  \begin{widetext}
\begin{equation*}
  \begin{split}
    &\left[  \frac{ 2  \sigma\left( \phi(\theta, \check{\theta})  \right)  \sigma\left(  \phi(\theta_0, \check{\theta}) \right)  }{ \sigma^2\left( \phi(\theta, \check{\theta})  \right) +  \sigma^2\left( \phi(\theta_0, \check{\theta})  \right)   } \right]^{1/2} < 1 \\
     &\iff  2  \sigma \left( m(\theta, \check{\theta}) \right)  \sigma \left( m(\theta_0, \check{\theta})  \right)    <  \sigma^2\left( m(\theta, \check{\theta})  \right) +  \sigma^2\left(  m(\theta_0, \check{\theta})  \right)\\
    &\iff \left(   \sigma\left( \phi(\theta, \check{\theta})  \right) -  \sigma\left(  \phi(\theta_0, \check{\theta})  \right) \right)^2 > 0,
\end{split}
\end{equation*}
 \end{widetext}
which holds from the identifiability condition.
\end{proof}

\begin{lemma}
  \label{obs2}
  For any $\theta \in \tilde{\Theta}$ with $\theta \neq \theta_0$, it follows
  that:
\begin{equation}
  \label{eq:lemma2}
  g(\theta) = \sup_{\check{\theta} \in \tilde{\Theta}} \mathrm{E}_{\theta_0}\left[ e^{\left(   R(\theta, \theta_0, \check{\theta}) \right)^{1/2}} \right] < 1.
\end{equation}
\end{lemma}

\begin{proof}
\begin{widetext}
  \begin{equation*}
   \begin{split}
     g(\theta) &=  \sup_{\check{\theta} \in \tilde{\Theta}} \int_{\mathbb{R}^{\nu}} \sqrt{  f\left(\vec{x}_{\nu}(\check{\theta}) \mid \phi(\theta, \check{\theta}) \right) f\left(\vec{x}_{\nu}(\check{\theta}) \mid  \phi(\theta_0, \check{\theta})   \right) } \, d\vec{x}_{\nu}(\check{\theta})\\
               &=\sup_{\check{\theta} \in \tilde{\Theta}} \int_{\mathbb{R}^{\nu}} \frac{ e^{ -\sum_{j=1}^{\nu}\frac{x_j^2}{4} \frac{ \sigma^2\left( \phi(\theta, \check{\theta})\right) + \sigma^2\left( \phi(\theta_0, \check{\theta})\right)  }{ \sigma^2\left( \phi(\theta, \check{\theta})\right) \sigma^2\left( \phi(\theta_0, \check{\theta})\right)  } }}{ \left[    2\pi\sigma\left( \phi(\theta, \check{\theta})\right) \sigma\left( \phi(\theta_0, \check{\theta})  \right) \right]^{\nu/2} } d \vec{x}_{\nu}(\check{\theta})\\
               &= \sup_{\check{\theta} \in \tilde{\Theta}}  \left[  \frac{2 \sigma\left(\phi(\theta_0, \check{\theta}) \right) \sigma\left(\phi(\theta, \check{\theta})\right) }{ \sigma^2\left(\phi(\theta, \check{\theta})\right) + \sigma^2\left(  \phi(\theta_0, \check{\theta}) \right) }   \right]^{\nu/2}.
   \end{split}
\end{equation*}
\end{widetext}
Since $\tilde{\Theta}$ is compact, the set $\tilde{\Theta}$ contains its
supremum. Then there exits a point $\check{\theta}_{\text{sup}} \in
\tilde{\Theta}$ such that
\begin{widetext}
\begin{equation*}
  \sup_{\check{\theta} \in
    \tilde{\Theta}}  \left[  \frac{2 \sigma\left( \phi(\theta, \check{\theta})   \right)
      \sigma\left(\phi(\theta_0, \check{\theta})\right) }{ \sigma^2\left( \phi(\theta, \check{\theta})  \right) +
      \sigma^2\left(  \phi(\theta_0, \check{\theta})  \right) }   \right]^{\nu/2} =  \left[  \frac{2 \sigma\left(  \phi(\theta, \check{\theta}_{\text{sup}})  \right)
      \sigma\left(  \phi(\theta_0, \check{\theta}_{\text{sup}})  \right)   }{  \sigma^2\left(  \phi(\theta, \check{\theta}_{\text{sup}})  \right)  +
      \sigma^2\left( \phi(\theta_0, \check{\theta}_{\text{sup}})  \right)   }   \right]^{\nu/2},
\end{equation*}
\end{widetext}
with $\tilde{\theta}^{\text{sup}} = \theta + \theta_{\text{opt}} -
\check{\theta}_{\text{sup}}$ and $\tilde{\theta}^{\text{sup}}_0 = \theta_0 +
\theta_{\text{opt}} - \check{\theta}_{\text{sup}}$. Therefore the proof of
Eq.~\eqref{eq:lemma2} follows from Eq.~\eqref{eq:lemma1} in \textbf{Lemma}
\textbf{\ref{obs1}}.
\end{proof}

\begin{lemma}
  \label{obs3}
  Let $ \bar{S}^{\epsilon}(\check{\theta}) = \mathrm{E}_{\theta_0}\left[
    e^{R(\epsilon, \theta, \theta_0, \check{\theta})^{1/2}} \right]$ and
  $\bar{S}(\check{\theta}) = \mathrm{E}_{\theta_0}\left[ e^{R(\theta, \theta_0,
      \check{\theta})^{1/2}} \right]$ for any $\theta \in \tilde{\Theta}$ with
  $\theta \neq \theta_0$. Then $\lim_{\epsilon \downarrow 0}
  \bar{S}^{\epsilon}(\check{\theta}) = \bar{S}(\check{\theta})$, and the
  following inequality holds:
 \begin{equation}
   \label{eq:lemma3}
   \lim_{\epsilon \downarrow 0}  \sup_{\check{\theta} \in \tilde{\Theta}} \bar{S}^{\epsilon}(\check{\theta}) < 1.
 \end{equation}
\end{lemma}

\begin{proof}
  Given that
  \begin{equation*}
    \scalemath{0.785}{\bar{S}^{\epsilon}(\check{\theta}) =  \int_{\mathbb{R}^{\nu}} \sup_{\theta' \in N_{\epsilon}(\theta) } \sqrt{ f\left(\vec{x}_{\nu}(\check{\theta}) \mid \phi(\theta', \check{\theta}) \right) f\left(\vec{x}_{\nu}(\check{\theta}) \mid  \phi(\theta_0, \check{\theta}) \right)} d\vec{x}_{\nu}(\check{\theta}),}
  \end{equation*}
  by the continuity of the density function, it follows that $\lim_{\epsilon \downarrow 0}
  \bar{S}^{\epsilon}(\check{\theta}) = \bar{S}(\check{\theta})$ for any $\theta
  \in \tilde{\Theta}$. Moreover, the closure of $N_{\epsilon}(\theta)$,
  $\overline{N_{\epsilon}(\theta)} \subset \tilde{\Theta}$ is compact, because
  $\tilde{\Theta}$ itself is a compact set. Consequently,
  \begin{equation*}
    \scalemath{0.785}{\bar{S}^{\epsilon}(\check{\theta}) =  \int_{\mathbb{R}^{\nu}} \max_{\theta' \in \overline{N_{\epsilon}(\theta)} }\sqrt{ f\left(\vec{x}_{\nu}(\check{\theta}) \mid \phi(\theta', \check{\theta}) \right) f\left(\vec{x}_{\nu}(\check{\theta}) \mid \phi(\theta_0, \check{\theta}) \right) }d\vec{x}_{\nu}(\check{\theta}).}
  \end{equation*}

  We note that the function $$\sqrt{ f\left(\vec{x}_{\nu}(\check{\theta}) \mid
      \phi(\theta', \check{\theta}) \right) f\left(\vec{x}_{\nu}(\check{\theta})
      \mid \phi(\theta_0, \check{\theta}) \right) }$$ is continuous over the
  Cartesian product of the compact sets $\overline{N_{\epsilon}(\theta)} \times
  \tilde{\Theta}$, which is also compact. Then the integrand in
  $\bar{S}^{\epsilon}(\check{\theta})$ is a continuous function at
  $\check{\theta}$ \cite{rudin1965,Fujiwara2006}. Therefore, the sequence of
  functions $\left( \bar{S}^{\epsilon}(\check{\theta}) \right)_{\epsilon
    \downarrow 0}$ forms a monotonically decreasing sequence of continuous
  functions defined on the compact set $\tilde{\Theta}$.

  By Dini's theorem \cite{rudin1965}, the convergence from
  $\bar{S}^{\epsilon}(\check{\theta}) \to \bar{S}(\check{\theta})$ as $\epsilon
  \downarrow 0$ is uniform in $\theta$. As a consequence,
  \begin{equation*}
    \lim_{\epsilon \downarrow 0}  \sup_{\theta \in \tilde{\Theta}} \bar{S}^{\epsilon}(\check{\theta}) =   \sup_{\theta \in \tilde{\Theta}} \lim_{\epsilon \downarrow 0} \bar{S}^{\epsilon}(\check{\theta}) = \sup_{\theta \in \tilde{\Theta}} \bar{S}(\check{\theta}) = g(\theta),
  \end{equation*}
  which, according to Eq.~\eqref{eq:lemma2} in \textbf{Lemma \ref{obs2}}, is
  less than $1$.
\end{proof}

\begin{lemma}
  \label{obs4}
  For any $\theta \in \tilde{\Theta}$ with $\theta \neq \theta_0$, there exist
  $\epsilon > 0$ and $b > 0$ such that for any $m \in \mathbb{N}$, and any set
  of homodyne measurements parameterized by $\left\{ \check{\theta}_i
  \right\}_{1 \leq i \leq m} \subset \tilde{\Theta}$, the following inequality
  holds
  \begin{equation}
    \label{eq:lemma4}
    P_{\theta_0}\left( \sum_{i=1}^m R(\epsilon, \theta, \theta_0,  \check{\theta}_i) > 0 \right) \leq e^{-bm}, \, m \geq 1.
  \end{equation}
\end{lemma}

\begin{proof}
  We start by observing that
  \begin{equation*}
    \label{eq:lem_markov}
    \scalemath{0.92}{
    \begin{split}
      P_{\theta_0}\left[  \sum_{i=1}^m R(\epsilon, \theta, \theta_0,  \check{\theta}_i) > 0 \right] &= P_{\theta_0}\left[  e^{\sum_{i=1}^m R(\epsilon, \theta, \theta_0,  \check{\theta}_i)} > 1 \right] \\
                                                                                                    &= P_{\theta_0}\left[  e^{\frac{1}{2}\sum_{i=1}^m R(\epsilon, \theta, \theta_0,  \check{\theta}_i)} > 1 \right] \\
                                                                                                    &= P_{\theta_0}\left[ \prod_{i=1}^{m} e^{ R(\epsilon, \theta, \theta_0,  \check{\theta}_i)/2} > 1 \right].
    \end{split}}
  \end{equation*}
  Applying the Markov's inequality, we obtain
  \begin{equation} \label{eq:markov}
    P_{\theta_0}\left[ \prod_{i=1}^{m} e^{ R(\epsilon, \theta, \theta_0,  \check{\theta}_i)/2} > 1 \right] \leq \mathrm{E}_{\theta_0}\left[ \prod_{i=1}^{m} e^{ R(\epsilon, \theta, \theta_0,  \check{\theta}_i)/2} \right].
  \end{equation}
  We note that Eq.~\eqref{eq:lemma3} implies that there exist
  sufficiently small $\epsilon > 0$ and $b > 0$ such that
\begin{equation}
  \label{eq:lemma4_II}
\sup_{\check{\theta} \in \tilde{\Theta}} \bar{S}^{\epsilon}(\check{\theta})  = e^{-b} < 1.
\end{equation}

Let $Z_i = e^{ R(\epsilon, \theta, \theta_0, \check{\theta}_i)/2}$ for $1 \leq i
\leq m$, and $Z_0 = 1$. We can define a stochastic process
\begin{equation*}
Y_m = \prod_{i=0}^{m} Z_m, \, m \geq 0.
 \end{equation*}
 This sequence of independent and non-negative random variables $\left\{ Y_m
 \right\}_{m \geq 0}$ forms an adapted stochastic process relative to the
 filtered space $\left( \Omega, \mathcal{F}, \left\{ \mathcal{F}_m \right\}, P
 \right)$, where $\left\{ \mathcal{F}_m; m \geq 0 \right\}$ is the natural
 filtration, $\mathcal{F}_0 := \left( \emptyset, \Omega \right)$, and
 $\mathcal{F}_m := \sigma\left( \vec{X}_{\nu}(\check{\theta}_1)
   ,\vec{X}_{\nu}(\check{\theta}_2),\ldots,\vec{X}_{\nu}(\check{\theta}_m)
 \right)$ \cite{Williams1991}. Furthermore, according to Eq.~\eqref{eq:lemma1},
 this stochastic process decreases on average, almost surely (a.s.), as
 indicated by
\begin{equation*}
  \mathrm{E}_{\theta_0}\left[ Z_m \mid \mathcal{F}_{m-1} \right] \leq Z_{m-1}, \, m \geq 1,
\end{equation*}
where the conditional expectation is well defined, since $Y_m$ is a
$\mathcal{F}_m$-measurable function. Consequently, the sequence $\left\{ Y_m
\right\}_{m \geq 0}$ consists of non-negative independent random variables that
satisfies the supermartingale condition \cite{Williams1991}. Hence by the tower
property of supermartingales, together with Eq.~\eqref{eq:lemma4_II}, for $m
\geq 1$ and for a sufficiently small $\epsilon > 0$, we have
\begin{equation*}
  \begin{split}
    \mathrm{E}_{\theta_0}\left[ Y_m \mid \mathcal{F}_{m-1} \right] &= \mathrm{E}_{\theta_0}\left[ Y_{m-1}Z_m \mid \mathcal{F}_{m-1} \right]\\
                                                                   &= Y_{m-1} \mathrm{E}_{\theta_0}\left[ Z_m \mid \mathcal{F}_{m-1} \right]\\
                                                                   &\leq Y_{m-1} e^{-b}.
 \end{split}
\end{equation*}
The iteration of this expectation through the filtration levels yields:
\begin{equation*}
\mathrm{E}_{\theta_0}\left[ Y_m \right] = \mathrm{E}_{\theta_0}\left[ Y_m \mid \mathcal{F}_{m-1} \mid \mathcal{F}_{m-2} \mid \cdots \mid \mathcal{F}_0 \right] \leq e^{-bm}.
\end{equation*}
Finally, incorporating this inequality into Eq.~\eqref{eq:markov}, we conclude
the proof.
\end{proof}

\begin{theorem}[Asymptotic strong consistency] \label{MainTheorem} Let $\theta_0
  \in \tilde{\Theta}$ be the true value of the phase, and be an interior point
  of $\tilde{\Theta}$. Consider $\left\{
    \widehat{\theta}(\vec{X}_{\nu}\left(\check{\theta}_1),\ldots,\vec{X}_{\nu}(\check{\theta}_m)\right)
  \right\}_{1 \leq i \leq m}$ a sequence of MLEs over $m$ adaptive steps,
  defined by AQSE over the set of homodyne measurements $\left\{
    M_{\check{\theta}_i}(dx) \right\}_{1 \leq 1 \leq m}$. Then,
  \begin{equation}
    \label{eq:theom_1}
    \widehat{\theta}\left(\vec{X}_{\nu}(\check{\theta}_1),\ldots,\vec{X}_{\nu}(\check{\theta}_m)\right)  \xrightarrow{\text{a.s.}}  \theta_0 \text{ as } m \to \infty.
  \end{equation}
\end{theorem}

\begin{proof}
  Let
  \begin{equation*}
    S_m(\theta_0, \theta) = \sum_{i=1}^{m} R(\theta, \theta_0,  \check{\theta}_i)
  \end{equation*}
  and
  \begin{equation*}
    S^{\epsilon}_m(\theta_0, \theta) = \sum_{i=1}^{m}R(\epsilon, \theta, \theta_0, \check{\theta}_i).
  \end{equation*}
  For any $a \geq 0$, let
  \begin{equation*}
    N_{a}(\theta_0)^{c} = \left\{ \theta \in \tilde{\Theta} : \lvert \theta -
      \theta_0 \rvert \geq a \right\}.
  \end{equation*}
  Given that $\tilde{\Theta}$ is compact, then $N_{a}(\theta_0)^{c}$ is also
  compact. Then, for every arbitrary collection $\mathcal{K}$ of open intervals of
  $\tilde{\Theta}$, such that
  $$N_{a}(\theta_0)^{c} \subseteq \cup_{K \in \mathcal{K}}
  K,$$ there exits a finite subcollection $\mathcal{J} \subseteq \mathcal{K}$ such
  that
  $$N_{a}(\theta_0)^{c} \subseteq \cup_{J \in \mathcal{J}} J.$$
  Each $J \in \mathcal{J}$ can be represented as a neighborhood $N_{\epsilon_i}(\theta_i)$
  for some $\theta_i \in \tilde{\Theta}$ and $\epsilon_i > 0$. Therefore, for
  some finite set $\left\{ \theta_1, \ldots, \theta_j \right\}$,
  $$N_{a}(\theta_0)^{c} \subseteq \cup_{\theta \in \left\{ \theta_1, \ldots,
      \theta_j \right\}} N_{\epsilon_i}\left( \theta_i \right),$$ where the
  $\epsilon_i$'s are chosen such that Eq.~\eqref{eq:lemma4} in \textbf{Lemma
    \ref{obs4}} holds for every $\theta_i$ and its corresponding $b_i$, with $i
  = 1,\ldots,j$.

  The next step in this proof consists of bounding the probability that the
  estimator $ \widehat{\theta}$ after $m$ adaptive measurements is at a distance
  larger than or equal to $a$ from the true value $\theta_0 $, i.e.,
$$\Big\lvert
\widehat{\theta}\left(\vec{X}_{\nu}(\check{\theta}_1)
,\ldots,\vec{X}_{\nu}(\check{\theta}_m) \right) - \theta_0 \Big\rvert \geq a.$$ This
occurs when the maximum of the likelihood function $L_m(\theta) =
\prod_{i=1}^{m} f( \vec{X}_{\nu}(\check{\theta}_i) \mid \phi(\theta, \check{\theta}_i) )$
belongs to $N_a(\theta_0)^c$. This would imply that $L_m(\widehat{\theta}) >
L(\theta_0)$ and hence $$ \sup_{\theta \in N_{a}(\theta_0)^c} \prod_{i=1}^{m}
\frac{ f\left( \vec{X}_{\nu}(\check{\theta}_i) \mid \phi(\theta, \check{\theta}_i) \right)}{ f\left(
  \vec{X}_{\nu}(\check{\theta}_i) \mid \phi(\theta_0, \check{\theta}_i) \right) } > 1.$$ Thus
  \begin{widetext}
  \begin{equation*}
    \begin{split}
      P_{\theta_0}\left(  \Big\lvert \widehat{\theta}\left(\vec{X}_{\nu}(\check{\theta}_1) ,\ldots,\vec{X}_{\nu}(\check{\theta}_m) \right) - \theta_0 \Big\rvert  \geq a \right) 
                                                                                            &=P_{\theta_0}\left( \sup_{\theta \in N_{a}(\theta_0)^c} \prod_{i=1}^m
                                                                                              \left[ \frac{ f\left(\vec{X}_{\nu}(\check{\theta}_i) \mid \phi(\theta, \check{\theta}_i) \right)  }{f\left(\vec{X}_{\nu}(\check{\theta}_i) \mid \phi(\theta_0, \check{\theta}_i) \right)   } \right] > 1 \right)\\
                                                                                            &= P_{\theta_0}\left( \sup_{\theta \in N_{a}(\theta_0)^c} \sum_{i=1}^m
                                                                                              \log \left[ \frac{ f\left(\vec{X}_{\nu}(\check{\theta}_i) \mid \phi(\theta, \check{\theta}_i) \right)  }{f\left(\vec{X}_{\nu}(\check{\theta}_i)  \mid \phi(\theta_0, \check{\theta}_i) \right)  } \right] > 0 \right) \\
                                                                                            &= P_{\theta_0}\left(  \sup_{\theta \in N_{a}(\theta_0)^c}  S_m(\theta_0,\theta) > 0   \right) \\ &  \leq P_{\theta_0}\left( \max_{1 \leq i \leq j} S_m^{\epsilon_i}(\theta_0,\theta_i) > 0 \right)\\
                                                                                            &\leq j \exp\left[ - \min_{1 \leq i \leq j}[b_i] \, m \right], \, m \geq 1.
   \end{split}
 \end{equation*}
\end{widetext}
In the third step of the previous sequence of equations, we have used the fact
that the logarithm is a strictly increasing function, and that $\tilde{\Theta}$
is compact, to commute the supremum with the logarithm. To bound the probability
$P_{\theta_0}$ in the last step of the sequence of equations, we have used
Eq.~\eqref{eq:lemma4} in \textbf{Lemma \ref{obs4}}.

Finally, by the Borell Cantelli lemma \cite{Chandra}, we conclude that
\begin{equation*}
  P_{\theta_0}\left( \Big\lvert \widehat{\theta}\left(\vec{X}_\nu(\check{\theta}_1) ,\ldots,\vec{X}_\nu(\check{\theta}_m) \right)  - \theta_0 \Big\rvert \geq a \quad \text{i.o.}  \right) = 0,
\end{equation*}
were $\text{i.o.} $ stands for \textit{infinitely often}. Therefore,
$\widehat{\theta}\left(\vec{X}_\nu(\check{\theta}_1) ,\ldots,\vec{X}_\nu(\check{\theta}_m) \right)
\xrightarrow{a.s.} \theta_0$ as $m \to \infty$.
\end{proof}

\begin{corollary}
  $\widehat{\theta}\left(\vec{X}_{\nu}(\check{\theta}_1) ,\ldots,\vec{X}_{\nu}(\check{\theta}_m) \right) =  \theta_0
 + o_{P}(1)$.
\end{corollary}

\begin{proof}
  Let $(W_m)_{m \geq 1}$ be a sequence of random variables. If $W_m = o_p(1)$,
  then the stochastic sequence $(W_m)_{m \geq 1}$ converges in probability to
  $0$ \cite{bishop2007discrete}. It implies that with arbitrary high
  probability, $\lvert W_m \rvert = o(1)$. Specifically, for any $\epsilon,
  \delta > 0$, there exists $N_{0}(\epsilon, \delta)$, such that for any $m
  >N_{0}(\epsilon, \delta)$,
  \begin{equation}
    P\left( \lvert W_m \rvert < \epsilon \right) \geq 1-\delta.
  \end{equation}
  On the other hand, Eq.~\eqref{eq:theom_1} in \textbf{Theorem
    \ref{MainTheorem}} guarantees the almost sure convergence for the sequence
  of MLEs $\left(   \widehat{\theta}\left(\vec{X}_{\nu}(\check{\theta}_1)
    ,\ldots,\vec{X}_{\nu}(\check{\theta}_m) \right) \right)_{m \geq 1}$ to $\theta_0$. Almost
  sure convergence implies convergence in probability \cite{Keener2010}.
  Therefore, we can conclude that $\widehat{\theta}\left(\vec{X}_{\nu}(\check{\theta}_1)
  ,\ldots,\vec{X}_{\nu}(\check{\theta}_m) \right) - \theta_{0}= o_P(1)$ as the number of adaptive
  steps $m$ increases.
\end{proof}

This means that the sequence of MLEs
$\left(\widehat{\theta}(\vec{X}_{\nu}(\check{\theta}_1)
  ,\ldots,\vec{X}_{\nu}(\check{\theta}_m) ) \right)_{m \geq 1}$ gets arbitrarily
close to $\theta_0$ with increasing probability as $m$ becomes very large
\cite{bishop2007discrete}.


 \subsubsection{\label{Appendix_AsypmNor} Saturation of the QCRB}

 A consequence of the \textit{almost sure convergence} of the MLE described in
 Appendix~\ref{Appendix_Proof_s1} is that the distribution of the sequence of
 MLEs converges to a normal distribution (asymptotic normality) with variance
 equal to the inverse of the QFI as the number of adaptive steps tends to
 infinity. To prove this statement we follow a similar methodology as in the
 case of i.i.d. random variables \cite{Keener2010}. However, in our case, the
 measurement samples $\vec{X}_{\nu}(\check{\theta}_1), \ldots,
 \vec{X}_{\nu}(\check{\theta}_m)$ are not identically distributed. Moreover,
 this proof requires the additional assumption that the likelihood functions
 $f(\vec{X}_{\nu}(\check{\theta}) \mid \phi(\theta, \check{\theta}) )$ are
 sufficiently smooth and continuous on $\mathbb{R}^{\nu} \times \tilde{\Theta}
 \times \tilde{\Theta}$. This assumption ensures that the existence of the
 Fisher information at every point in $\mathbb{R}^{\nu} \times \tilde{\Theta}
 \times \tilde{\Theta}$.

 Under the regularity conditions described in Appendix~\ref{sec:MLE} and with
 the additional assumption of sufficient smoothness of the likelihood functions,
 we expand the derivative of the logarithm of the likelihood $l_m(\theta) =
 \sum_{i=1}^{m} \log \left[ f(\vec{X}_{\nu}(\check{\theta}_i) \mid \phi(\theta,
   \check{\theta}_i) ) \right]$ about $\theta_0$ up to second order, yielding
\begin{equation}
  \label{eq:as_nor}
  l'_m(\theta) = l'_m(\theta_0) + l''_m(\theta_0)(\theta - \theta_0) + \frac{1}{2}l_m'''(\theta^*)(\theta - \theta_0)^2,
\end{equation}
where $\theta^* \in \tilde{\Theta}$ such that $|\theta^* - \theta_0| < |\theta -
\theta_0|$. Here, the regularity conditions in Appendix:~\ref{sec:MLE} ensure
that the estimates from $\widehat{\theta}(\vec{X}_{\nu}(\check{\theta}_1)
,\ldots,\vec{X}_{\nu}(\check{\theta}_m) )$ are stationary points of the
likelihood function. Then, when evaluating Eq.~\eqref{eq:as_nor} at $\theta =
\widehat{\theta}(\vec{X}_\nu(\check{\theta}_1)
,\ldots,\vec{X}_{\nu}(\check{\theta}_m) )$, the left-hand side of this equation
vanishes, yielding
\begin{widetext}
\begin{equation}
  \label{eq:as_nor_2}
  \sqrt{m}\left( \widehat{\theta}(\vec{X}_\nu(\check{\theta}_1)
    ,\ldots,\vec{X}_\nu(\check{\theta}_m)) - \theta_0   \right) = \frac{ \frac{ l_m'(\theta_0)}{\sqrt{m}}   }{ - \frac{l_m''(\theta_0)}{m} - \frac{l'''_m(\theta^*)}{2m} \left( \widehat{\theta}(\vec{X}_\nu(\check{\theta}_1)
      ,\ldots,\vec{X}_\nu(\check{\theta}_m) ) - \theta_0 \right) }.
\end{equation}
\end{widetext}
Moreover, given that the MLE is asymptotically consistent, i.e.,
$\widehat{\theta}\left(\vec{X}_\nu(\check{\theta}_1)
  ,\ldots,\vec{X}_\nu(\check{\theta}_m)\right) \xrightarrow{\text{a.s.}}
\theta_0$, and the smoothness of the likelihood, it can be shown that
\cite{Fujiwara2006}:
\begin{itemize}
\item[(a)] $l'_m(\theta_0)/\sqrt{m} \to \mathcal{N}(0,
  1/\nu F_X(\theta_{\text{opt}}))$ in distribution,
\item[(b)] (b) $l_m''(\theta_0)/m \to
  - \nu F_X(\theta_{\text{opt}})$ in probability, and
\item[(c)] $l'''_m(\theta^*)/m$ is bounded
  in probability.
\end{itemize}

This results in the asymptotic normality of the estimator:
\begin{equation}
	\label{eq:asymptotic_normal_3}
	\sqrt{m} \left(  \widehat{\theta}\left(   \vec{X}_\nu(\check{\theta}_1) ,\ldots,\vec{X}_\nu(\check{\theta}_m) \right) \right)  \xrightarrow{\text{d}} \mathcal{N}\left( \theta_0, \frac{1}{ \nu F_Q} \right),
\end{equation}
where the $d$ above the arrow denotes convergence in distribution. Therefore,
the distribution of the estimator follows a normal distribution with mean equal
to $\theta_0$ and a variance equal to the inverse of the QFI, and therefore
saturates the QCRB $[N F_{X}(\theta_{\text{opt}})]^{-1}$ in the asymptotic limit
of $m \to \infty$, with $N=\nu \times m$.

We can give the bound for the rate of convergence using asymptotic normality.
Similar to how the definition of $o_p$ states the convergence in probability, if
a sequence of random variables $W_m = O_p(a_m)$, then with high probability,
$|W_m| = O(a_m)$ (i.e. for sufficiently large $m$, the sequence $W_m$ is bounded
above by a constant multiple of the sequence $a_m$ in probability)
\cite{bishop2007discrete}. Specifically, for every $\epsilon >0$ there exists a
constant $K(\epsilon)$ and an integer $N_{0}(\epsilon)$ such that if $m >
N_{0}(\epsilon)$, then
\begin{equation}
  \scalemath{1.0}{P\left( \left| \frac{W_m}{a_m} \right| \leq K(\epsilon) \right) \geq 1-\epsilon.}
\end{equation}
Since Eq.~\eqref{eq:asymptotic_normal_3} holds for $W_m =
\widehat{\theta}-\theta_0$, we can apply the Chebyshev inequality
\cite{Keener2010}. This inequality states that for any $\epsilon > 0$, we can
choose a constant $K(\epsilon) = 1/\sqrt{\epsilon} > 0$ such that for a
sufficiently large $m$,
\begin{equation}
  \scalemath{0.92}{P_{\theta}\left( \left| \sqrt{m} \left( \widehat{\theta}\left(\vec{X}_\nu(\check{\theta}_1), \ldots, \vec{X}_\nu(\check{\theta}_m)\right) - \theta_0 \right) \right| < K(\epsilon) \right) \geq 1-\epsilon.}
\end{equation}
Therefore, we find that $\widehat{\theta}\left(\vec{X}_\nu(\check{\theta}_1), \ldots,
\vec{X}_\nu(\check{\theta}_m)\right) - \theta_0 = O_p(1/\sqrt{m})$.


\subsection{\label{Appendix3} Phase estimation around $\theta_{opt}$ }

The estimator in Eq.~(\ref{eq:estimator_proposal}) of the proposed adaptive
homodyne strategy minimizes the variance over all the phases within the
parametric space from $[0,\pi/2)$. This estimator yields a variance above the
QCRB for phases close to $\theta_{opt}$ for a small number of adaptive steps
$m$. On the other hand, we note that the two-step protocol from Ref.
\cite{Berni2015} is closer to the QCRB for $\theta\approx\theta_{opt}$ than our
strategy for small $m=3,5$. However, our proposed strategy is more general, and
encompasses the one from Ref. \cite{Berni2015}. By taking an uneven splitting
ratio between the first and second step and $\check{\theta}_{1}=\theta_{opt},$
our strategy with $m=2$ reduces to the strategy from Ref. \cite{Berni2015}.
Nevertheless, by construction our strategy is guaranteed to achieve the QCRB for
any phase in the asymptotic limit.

\subsection{\label{AppendixExtension}Convergence of the CHH strategy over the
  full range $[0,\pi)$}

The asymptotic consistency of $\widehat{\theta}\left(\vec{X}_\nu(\check{\theta}_1)
  ,\ldots,\vec{X}_\nu(\check{\theta}_m)\right)$ can be extended to a compact parameter
space $\tilde{\Theta} \subset [0, \pi)$ by replacing the initial step of
sampling $\vec{X}_{\nu}(\check{\theta}_1)$ with a sample $\vec{X}_\nu^{(1)} \in
\mathbb{C}^{\nu}$ of size $\nu$ obtained from a series of heterodyne
measurements. This preliminary heterodyne sampling provides sufficient
information about $\theta_0$ to overcome the non-identifiability problem
inherent to homodyne measurements, thereby allowing for extending the parameter
space.

The proof of asymptotic consistency of the MLE and convergence of the estimator
variance to the QCRB in Appendix~\ref{Appendix_Proof} remains valid for the full
range $[0,\pi)$ if we can show a new version of \textbf{Lemma~\ref{obs2}}. This
new \textbf{Lemma} incorporates the likelihood function from the heterodyne
measurement as a multiplicative factor in the likelihood function for the
subsequent homodyne measurements. From this result, the remaining
\textbf{Lemmas}~\ref{obs3}, \ref{obs4} and \textbf{Theorem}~\ref{MainTheorem} in
Appendix~\ref{Appendix_Proof} for proving almost sure convergence of the MLE
follow analogously.

\begin{lemma}
  \label{obs_heterodyne}
  For any $\theta \in \tilde{\Theta}$ with $\theta \neq \theta_0$, it follows
  that:
  \begin{equation}
    \label{eq:lem_heterodyne}
    g_{\mathrm{CHH}}(\theta) = \sup_{\check{\theta} \in \tilde{\Theta}} \mathrm{E}_{\theta_0}\left[ e^{\left(   R_{\mathrm{CHH}}(\theta, \theta_0, \check{\theta}) \right)^{1/2}} \right] < 1,
  \end{equation}
  where
  \begin{equation*}
    \scalemath{0.95}{
      R_{\mathrm{CHH}}(\theta, \theta_0, \check{\theta}) = \log \left( \frac{ f\left(  \vec{X}_\nu^{\text{Het}}  \mid \theta\right)  f(\vec{X}_{\nu}(\check{\theta}) \mid \phi(\theta, \check{\theta}) )}{    f\left( \vec{X}_\nu^{\text{Het}}  \mid \theta\right) f(\vec{X}_{\nu}(\check{\theta}) \mid \phi(\theta_0, \check{\theta})   )} \right),}
 \end{equation*}
\end{lemma}

\begin{proof}
  Let $\vec{x}_\nu^{\text{Het}} = \alpha_1, \ldots, \alpha_\nu$ be an observed sample
  from $\vec{X}_\nu^{\text{Het}}  \in \mathbb{C^{\nu}}$ and assume $\theta \neq \theta_0$ for $\theta,
  \theta_0 \in \tilde{\Theta}$. In this case,
\begin{widetext}
\begin{equation}
  \label{eq:ext0}
  \begin{split}
    g_{\mathrm{CHH}}(\theta) &=  \sup_{\check{\theta} \in \tilde{\Theta}}   \int_{\mathbb{C}^\nu} \sqrt{ f\left( \vec{x}_\nu^{\text{Het}}  \mid \theta\right)  f\left( \vec{x}_\nu^{\text{Het}} \mid \theta_0\right) }   d  \vec{x}_\nu^{\text{Het}} \int_{\mathbb{R}^{\nu}} \sqrt{  f\left(\vec{x}_{\nu}(\check{\theta}) \mid \phi(\theta, \check{\theta}) \right) f\left( \vec{x}_{\nu}(\check{\theta}) \mid  \phi(\theta_0, \check{\theta})   \right) } \, d\vec{x}_{\nu}(\check{\theta})\\
                             &=   \int_{\mathbb{C}^\nu} \sqrt{ f\left( \vec{x}_\nu^{\text{Het}} \mid \theta\right)  f\left( \vec{x}_\nu^{\text{Het}} \mid \theta_0\right) }   d \vec{x}_\nu^{\text{Het}} \, \times g(\theta) ,
  \end{split}
\end{equation}
\end{widetext}
where $g(\theta)$ is defined in Eq.~\eqref{eq:lemma2}.

Since $g(\theta) < 1$ according to Lemma~\ref{obs2} in Appendix
\ref{Appendix_Proof}, then it suffices to show that:
 \begin{equation}
   \label{eq:ext1}
   \int_{\mathbb{C}^\nu} \sqrt{ f\left( \vec{x}_\nu^{\text{Het}} \mid \theta\right)  f\left( \vec{x}_\nu^{\text{Het}} \mid \theta_0\right) }   d \vec{x}_\nu^{\text{Het}}  < 1,
 \end{equation}
 For a heterodyne measurement outcomes $\vec{x}_\nu^{\text{Het}}$, the likelihood
 function is expressed as
\begin{equation}
  \label{eq:ext2}
  f( \vec{x}_\nu^{\text{Het}} \mid \theta) =  \left[ \pi \cosh(r)  \right]^{-\nu}  \prod_{j=1}^{\nu} e^{-|\alpha_j|^2 - \tanh(r) \mathrm{Re}\left[ \alpha_j^2 e^{\iu 2 \theta} \right] }.
\end{equation}
Then, the integral of the square root of the product of the likelihoods for
$\theta$ and $\theta_0$ becomes
\begin{widetext}
\begin{equation}
  \begin{split}
    &\int_{\mathbb{C}^\nu} \sqrt{ f\left( \vec{x}_\nu^{\text{Het}} \mid \theta\right)  f\left( \vec{x}_\nu^{\text{Het}} \mid \theta_0\right) }   d \vec{x}_\nu^{\text{Het}} =   \left[ \pi \cosh(r)  \right]^{-\nu}    \prod_{j=1}^{\nu} \int_{\mathbb{C}} e^{-  |\alpha_j|^2 - \frac{\tanh(r)}{2} ( \mathrm{Re}\left[ \alpha_j^2 e^{\iu 2 \theta} \right] - \mathrm{Re}\left[ \alpha_j^2 e^{\iu 2 \theta_0} \right]) } d\alpha_j \\
                                                                                                             &= \left[ \pi \cosh(r)  \right]^{-\nu}   \prod_{j=1}^{\nu} \int_{\mathbb{R}^2} e^{-  \left( a_j^2 + b_j^2 \right) - \tanh(r) \cos\left( \theta - \theta_0 \right) \left[    \cos(\theta + \theta_0)(a_j^2-b_j^2) - 2 a b \left( \sin(2 \theta) + \sin(2 \theta_0) \right)         \right] } da_jdb_j \\
    &= \left[\cosh(r) \sqrt{1- \cos^2(\theta-\theta_0)\tanh^2(r)} \right]^{-\nu}.
  \end{split}
\end{equation}
\end{widetext}
Since the term $\left[\cosh(r) \sqrt{1-
    \cos^2(\theta-\theta_0)\tanh^2(r)}\right]^{-1}$ <$ 1$ for $r > 0$, the
assertion follows.
\end{proof}

Since \textbf{Lemma \ref{obs_heterodyne}} holds, the proof of the asymptotic
consistency of the MLE and convergence to the QCRB proceeds analogously to the
proof in Appendix~\ref{Appendix_Proof}.

\subsubsection{Numerical tests of estimator normality for the CHH strategy  over $[0,\pi)$.}
\label{NormalTest}
Analogous to the case of adaptive homodyne in Section~\ref{sec:Op_dyne}, the MLE
from the CHH strategy is expected to show asymptotic normality, defined in
Eq.~\eqref{eq:asymptotic_normal_3}. We tested the normality of the MLE over
$[0,\pi)$. To this end, we conducted five independent Anderson–Darling
goodness-of-fit tests on $N = 1000$ samples of
$\sqrt{m}\left(\widehat{\theta}\left(
    \vec{X}_\nu^{\text{Het}},\vec{X}_{\nu}(\check{\theta}_2)
    ,\ldots,\vec{X}_{\nu}(\check{\theta}_m)\right)- \theta_0\right)$, with
$m=20$, and $\nu = 50$ for $5$ randomly selected values of $\theta_0$ within
$[0,\pi)$. The null hypothesis for each test posited that the samples from the
MLE follow a normal distribution $\mathcal{N}\left( \theta_0, \frac{1}{ 1000
    F_X(\theta_{\text{opt}})} \right)$ with $r = 1$. Using the \verb|DescTools|
package in \textbf{\textit{R}}, these tests produced a set of five p-values.
Subsequently, the Fisher method, implemented in the package of R \verb|poolr|,
calculated a combined p-value of $0.9150486$. Given this very high combined
p-value, we fail to reject the null hypothesis that the samples follow the
specified normal distribution. This study provides statistical evidence that the
MLE reaches the asymptotic normality with a moderate number of adaptive steps,
in this test $m=20$ steps.

\subsubsection{Degree of non-Gaussianity of the MLE for small $m$.}
\label{NonGaussianEffects}
In general, the distribution of MLEs obtained from the CHH strategy is a
non-Gaussian distribution for small $m$, and is expected to approach a normal
distribution as $m$ increases. We investigate the degree of non-Gaussianity of
the MLE for small $m$ by studying the first four central moments of the
estimator distribution from the CHH strategy as a function of $m$.
Figure~\ref{fig11:app} shows (a) the first moment (bias), (b) the second moment
(normalized Holevo variance), the third moment (skewness), and the fourth moment
(excess of kurtosis), for the CHH strategy with $N = 3000$ and $r$. We observe
that the bias quickly converges to zero, the variance approaches the QCRB, and
both the skewness and excess kurtosis tend to zero as $m$ increases. These
results show that while for a small $m$ the MLE shows non-Gaussian
characteristics, such as a small asymmetry (skewness) and the presence of
outliers (excess of kurtosis), the MLE tends to a normal distribution as $m$
increases. These results highlight the importance of having a sufficiently large
number of probe states $N$, and enough $m$, to ensure the convergence to the
QCRB for any $\theta \in [0, \pi)$.

\begin{figure}[ht]
  \centering
  \includegraphics[scale = 0.32]{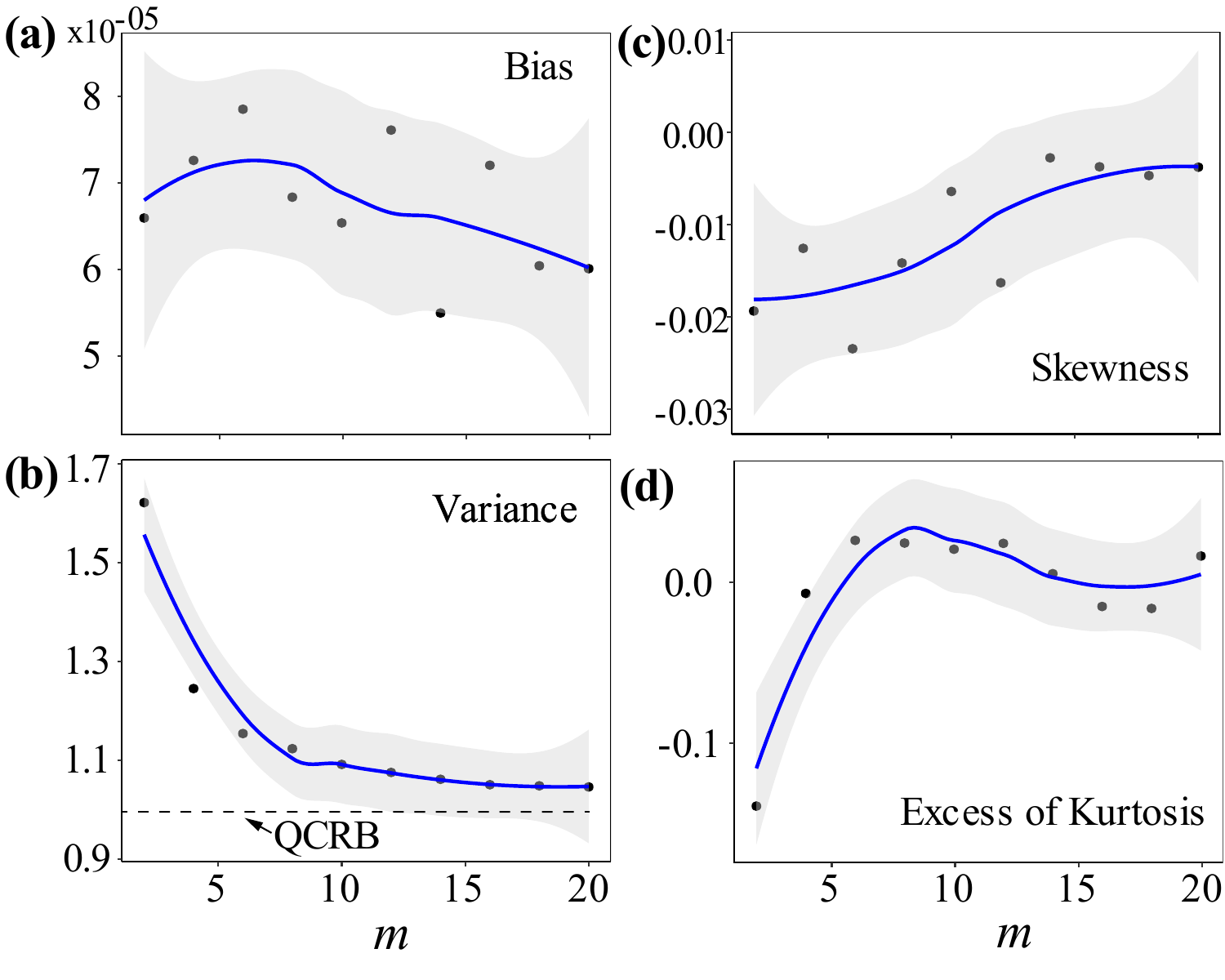}
  \caption{\label{fig11:app} Degree of non-Gaussianity of the MLE for the CHH
    strategy with $N = 3000$ and $r = 1$. Each panel represents the evolution of
    one of the central moments of the estimator distribution as a function of
    the number of adaptive steps $m$. \textbf{(a)} First moment (Bias),
    \textbf{(b)} Second moment (normalized Holevo variance), \textbf{(c)} Third
    moment (Skewness), and \textbf{(d)} Fourth moment (excess of kurtosis).
    Black dots represent the average of 5 Monte Carlo simulations, each with
    $1\times10^4$ samples. The blue lines represent the tendencies provided by
    the localized regression method with span of $0.9$, and the light-gray
    shadows represent the $95$ percent confidence interval. Note that the bias
    rapidly converges to zero and the variance approaches the QCRB. The skewness
    and excess of kurtosis tend to zero as $m$ increases.}
\end{figure}

\section*{Funding}
This work was funded by the National Science Foundation (NSF) Grants \#
PHY-2210447, FRHTP \# PHY-2116246, and the Q-SEnSE QLCI \# 2016244.

\section*{Acknowledgments}
We thank Laboratorio Universitario de Cómputo de Alto Rendimiento (LUCAR) of
IIMAS-UNAM for their service on information processing. We would like to thank
the UNM Center for Advanced Research Computing, supported in part by the
National Science Foundation, for providing the high performance computing
resources used in this work.

\section*{Disclosures}
The authors declare no conflicts of interest.

\section*{Data availability}
The data that support the findings of this study are available from the authors upon request.



\end{document}